\documentclass[12pt, draftclsnofoot, onecolumn]{IEEEtran}
\IEEEoverridecommandlockouts
\usepackage{url}
\usepackage{cite}
\usepackage{amsmath,amssymb,amsfonts}
\usepackage{graphicx}
\usepackage{textcomp}
\usepackage[utf8]{inputenc}
\usepackage{amssymb}
\usepackage{color}
\usepackage{relsize}
\usepackage{threeparttable}
\usepackage[nolist]{acronym}
\usepackage{multirow}
\usepackage{framed}
\usepackage{soul}
\usepackage{algpseudocode}
\usepackage{algorithm}
\usepackage{amsmath}
\usepackage{amsthm}
\usepackage[squaren,Gray]{SIunits}
\usepackage[dvipsnames]{xcolor}
\usepackage{balance}
\usepackage{comment}
\usepackage{longtable}
\usepackage{graphicx}
\usepackage{afterpage}

\ifCLASSOPTIONcompsoc
    \usepackage[caption=false, font=normalsize, labelfont=sf, textfont=sf]{subfig}
\else
\usepackage[caption=false, font=footnotesize]{subfig}
\fi

\title{\huge Spatial Multiplexing in Near Field MIMO Channels with Reconfigurable Intelligent Surfaces}

\newcommand{\kt}{N}
\newcommand{\kr}{M}
\newcommand{\kris}{K}

\newcommand{\Hf}{{\HH}_{\text{F}}}
\newcommand{\Hb}{{\HH}_{\text{B}}}
\newcommand{\Uf}{{\mathbf{U}}_{\text{F}}}
\newcommand{\Ub}{{\mathbf{U}}_{\text{B}}}
\newcommand{\Vf}{{\mathbf{V}}_{\text{F}}}
\newcommand{\Vb}{{\mathbf{V}}_{\text{B}}}
\newcommand{\Sf}{{\mathbf{\Sigma}}_{\text{F}}}
\newcommand{\Sb}{{\mathbf{\Sigma}}_{\text{B}}}

\newcommand{\nondiagopt}[1]{NDRO}
\newcommand{\focusing}[1]{FOC}
\newcommand{\optimization}[1]{DROP}

\newcommand{\Ptx}{\mathbf{P}_{{\text{T}}}}
\newcommand{\Prx}{\mathbf{P}_{{\text{R}}}}
\newcommand{\Prr}{\mathbf{P}_{{\text{RF}}}}
\newcommand{\Prt}{\mathbf{P}_{{\text{RB}}}}

\newcommand{\Ftx}{\mathbf{F}_{{\text{T}}}}
\newcommand{\Frx}{\mathbf{F}_{{\text{R}}}}
\newcommand{\Frr}{\mathbf{F}_{{\text{RF}}}}
\newcommand{\Frt}{\mathbf{F}_{{\text{RB}}}}

\newcommand{\B}[1]{\mathbf{{#1}}}
\newcommand{\T}[1]{\widetilde{{#1}}}
\newcommand{\D}[1]{\overline{{#1}}}
\newcommand{\HH}{\mathbf{H}}
\newcommand{\PHI}{\mathbf{\Phi}}
\newcommand{\Hphi}{\HH}
\newcommand{\herm}{\dag}
\newcommand{\Sfi}{{\sigma}_{\text{F}i}}
\newcommand{\Sbi}{{\sigma}_{\text{B}i}}

\newcommand{\diag}{{\rm diag}}

\newtheorem{lemma}{Lemma}

\newcommand{\boldI} {{\bf I}}

\def\BibTeX{{\rm B\kern-.05em{\sc i\kern-.025em b}\kern-.08em
    T\kern-.1667em\lower.7ex\hbox{E}\kern-.125emX}}

\newtheorem{proposition}{Proposition}

\begin{document}

\author{\normalsize G.~Bartoli~\IEEEmembership{\normalsize Member,~IEEE}, A.~Abrardo~\IEEEmembership{\normalsize Senior~Member,~IEEE}, N.~Decarli~\IEEEmembership{\normalsize Member,~IEEE}, D.~Dardari~\IEEEmembership{\normalsize Senior~Member,~IEEE}, and M.~Di~Renzo~\IEEEmembership{\normalsize Fellow,~IEEE}
\thanks{Manuscript submitted Dec. 21, 2022. G. Bartoli and A. Abrardo are with University of Siena, Italy. (e-mail: giulio.bartoli@unisi.it, abrardo@diism.unisi.it). N. Decarli is with the Italian Research Council, Italy. (e-mail: nicolo.decarli@ieiit.cnr.it). D. Dardari is with University of Bologna, Italy. (e-mail: davide.dardari@unibo.it). M. Di Renzo is with Universit\'e Paris-Saclay, CNRS, CentraleSup\'elec (L2S), France. (e-mail: marco.di-renzo@universite-paris-saclay.fr). This work was supported, in part, by the Theory Lab, Central Research Institute, 2012 Labs, Huawei Technologies Co., Ltd., the European Commission through the H2020 ARIADNE project under grant agreement number 871464 and through the H2020 RISE-6G project under grant agreement number 101017011, and the Fulbright Foundation.} \vspace{-1.75cm}
}

\maketitle

\begin{abstract}
We consider a multiple-input multiple-output (MIMO) channel in the presence of a reconfigurable intelligent surface (RIS). Specifically, our focus is on analyzing the spatial multiplexing gains in line-of-sight and low-scattering MIMO channels in the near field. We prove that the channel capacity is achieved by diagonalizing the end-to-end transmitter-RIS-receiver channel, and applying the water-filling power allocation to the ordered product of the singular values of the transmitter-RIS and RIS-receiver channels. The obtained capacity-achieving solution requires an RIS with a non-diagonal matrix of reflection coefficients. Under the assumption of nearly-passive RIS, i.e., no power amplification is needed at the RIS, the water-filling power allocation is necessary only at the transmitter. We refer to this design of RIS as a linear, nearly-passive, reconfigurable electromagnetic object (EMO). In addition, we introduce a closed-form and low-complexity design for RIS, whose matrix of reflection coefficients is diagonal with unit-modulus entries. The reflection coefficients are given by the product of two focusing functions: one steering the RIS-aided signal towards the mid-point of the MIMO transmitter and one steering the RIS-aided signal towards the mid-point of the MIMO receiver. We prove that this solution is exact in line-of-sight channels under the paraxial setup. With the aid of extensive numerical simulations in line-of-sight (free-space) channels, we show that the proposed approach offers performance (rate and degrees of freedom) close to that obtained by numerically solving non-convex optimization problems at a high computational complexity. Also, we show that it provides performance close to that achieved by the EMO (non-diagonal RIS) in most of the considered case studies.
\end{abstract}

\begin{IEEEkeywords}
Multi-antenna systems, reconfigurable intelligent surfaces, spatial multiplexing, near field channels.
\end{IEEEkeywords}

\section{Introduction}
\label{Sec:Introduction}
The need for terabits wireless links and sub-milliseconds (tactile) time responses in next-generation wireless systems, e.g., to support emerging augmented reality and virtual reality applications based on three-dimensional holographic video representations \cite{iown}, exceeds the current capabilities of fifth-generation telecommunication standards. A possible solution to fulfill these requirements is the migration towards high frequency bands, notably the sub-terahertz and terahertz spectrum \cite{ISG-THz}, \cite{THZ-Survey-COMAG}. The terahertz band offers, in fact, tens-hundreds gigahertz of contiguous bandwidth and enables the coexistence with other regulated spectra. Due to the large available bandwidth, the use of terahertz frequencies may extend the quality of experience of fiber-optic systems to wireless links, both in terms of data rates and time responses.

The use of terahertz signals for wireless communications is, however, challenging for several reasons, including the severe propagation losses that drastically limit the communication coverage, and the lack of diffraction that makes wireless propagation occurs predominantly in line-of-sight \cite{DBLP:journals/comsur/HanWLCAKM22}. Notably, it is shown in \cite{Can-THZ} that ensuring a strong line-of-sight connection is essential for fulfilling the desired communication performance. The severe path-loss may be overcome and the availability of line-of-sight connections may be ensured with high probability through (i) a denser deployment of base stations (network densification); (ii) the deployment of relays to split long-range transmission links into multiple short-range transmission links, possibly avoiding blocking objects; and (iii) the use of highly-directive antenna arrays at the transmitter and receiver \cite{DBLP:journals/cm/AkyildizHN18}. These solutions, however, usually result in an increase of backhaul infrastructure, power consumption, and hardware complexity \cite{DBLP:journals/ojcs/RenzoNSDQLRPSZD20}. Recent national and supranational directives have set, on the other hand, strict targets in reducing the energy consumption, such as the European Union that aims to reduce the power consumption by 32.5\% before 2030 \cite{EC-EE}. Recent activities within telecommunication standards organizations, e.g., the 3rd Generation Partnership Project (3GPP), in addition, are focused on the analysis of new network nodes to offer blanket coverage in cellular networks, in those scenarios where the deployment of full-stack cells may not always be possible (e.g., no availability of backhaul) or may not be economically viable \cite{NC-Repeaters}.

In this context, two emerging technologies have recently gained major attention from the research community, and are suitable for enabling high rate, reduced hardware complexity, and high energy efficient communications, especially for transmission in the terahertz frequency band. The first technology is known as reconfigurable intelligent surface (RIS) \cite{DBLP:journals/jsac/RenzoZDAYRT20}, \cite{DBLP:journals/cm/WuZ20} and the second is known as holographic surface (HoloS) \cite{DBLP:journals/wc/HuangHAZYZRD20}, \cite{DBLP:journals/cm/DardariD21}. Both technologies capitalize on recent progresses in the field of dynamic (reconfigurable) metasurfaces, even though they serve a different purpose and have different requirements \cite{marco_di_renzo_2022_7430186}. An RIS is a relay-type surface whose main features are the low implementation complexity and the low power consumption, compared with conventional relay nodes, e.g., decode-and-forward, amplify-and-forward, and the more recent network-controlled repeaters \cite{Flamini-TAP}. Specifically, an RIS is capable of shaping the electromagnetic waves with no power amplification and analog-to-digital conversion, which reduces the hardware complexity, power consumption, and processing delay (time response). A HoloS is, on the other hand, a hybrid continuous-aperture multiple-input multiple-output (MIMO) type transceiver \cite{C-MIMO}, which is endowed with  digital signal processing and power amplification capabilities, but with a reduced number of radio frequency chains compared to legacy MIMO implementations \cite{WDM-Luca}, \cite{https://doi.org/10.48550/arxiv.2210.08616}. Notably, two focused industry specification groups were recently established on terahertz communications and intelligent surfaces to coordinate the research efforts on these technologies and to strengthen their synergy and complementary \cite{ISG-THz}, \cite{ISG-RIS}.

An RIS and a HoloS share the feature of being electrically-large, i.e., their geometric size is much larger than the wavelength. This property makes them especially suitable for enabling spatial multiplexing in line-of-sight and low-scattering MIMO channels, which have been extensively studied in the past \cite{DBLP:conf/wcnc/BohagenOO05} and are currently receiving major renewed attention fueled by emerging potential applications and technology advances \cite{DBLP:journals/cm/DoCPSLL21}. Specifically, two electrically large HoloS can support multimode transmission (spatial multiplexing) in line-of-sight channels, provided that the transmission range is within the radiative near field (Fresnel region) of each other \cite{Miller00}, \cite{Anna-Thaning-2003}, \cite{DBLP:journals/jsac/Dardari20}. Asymptotic analysis shows that the spatial multiplexing gain, i.e., the number of degrees of freedom, is proportional to the product of the apertures of the HoloS and is inversely proportional to the square of the product of the wavelength and the transmission distance. The deployment of electrically-large HoloS at terahertz frequencies is, therefore, considered a suitable option to realize terabits wireless links with tactile time responses. The deployment of electrically-large RIS in such an environment is, likewise, considered an essential technology for ensuring reconfigurable short-range multi-hop links without requiring the deployment of large numbers of HoloS, thus fulfilling the requirements of blanket coverage at low cost and high energy efficiency, while preserving tactile time responses, since the signals are processed in the electromagnetic domain \cite{DBLP:journals/cm/AkyildizHN18}. In addition, the deployment of RIS may enable spatial multiplexing gains in each line-of-sight hop, even though the line-of-sight direct link (between the transmit and receive HoloS) does not support spatial multiplexing, since operating in the far field \cite[Sec. 7.2.5]{tse_viswanath_2005}.

Despite the potential benefits of RIS-aided HoloS systems, especially in line-of-sight and low-scattering channels, the achievable spatial multiplexing gains and the optimal design for the precoding, decoding, and reflection coefficients of HoloS and RIS are an open problem. The available research works can be cast in two categories, as summarized in the next sub-section.

\subsection{State of the Art}
The first category includes papers on the analysis and design of transmission links between two HoloS. The most relevant research works for this paper include \cite{Miller00}, \cite{DBLP:journals/jsac/Dardari20},  \cite{DBLP:journals/access/DecarliD21}, \cite{WDM-Luca}. In \cite{Miller00}, the author computes the number of degrees of freedom between two HoloS in the Fresnel region under the paraxial setup. In addition, the optimal precoding functions are given in closed-form and are shown to be prolate spheroidal functions. In \cite{DBLP:journals/jsac/Dardari20}, the analysis in \cite{Miller00} is generalized to account for shorter transmission distances (or HoloS with a larger aperture). In \cite{DBLP:journals/access/DecarliD21}, the authors consider non-paraxial setups and provide an approximated expression for the degrees of freedom and for the eigenfunctions, which are shown to be focusing functions. In \cite{WDM-Luca}, the authors propose and analyze a wavenumber-division multiplexing scheme to reduce the complexity of optimal spatial multiplexing schemes based on the prolate spheroidal functions identified in \cite{Miller00}. A more comprehensive overview, including recent works on line-of-sight MIMO, e.g., \cite{DBLP:journals/cm/DoCPSLL21} and references therein, can be found in \cite{https://doi.org/10.48550/arxiv.2210.08616}. None of these papers consider the presence of RIS.

In the rest of this paper, we model the transmit and receive HoloS as two critically-sampled MIMO arrays, i.e., having antenna elements spaced by half-wavelength. This is in agreement with \cite{C-MIMO} and \cite{DBLP:journals/tsp/PizzoTSM22}, since the scope of this paper does not require to model the impact of surface waves \cite{https://doi.org/10.48550/arxiv.2210.08619}. Therefore, we refer to the considered system model as an RIS-aided MIMO channel.

The second category includes papers on the analysis and design of RIS-aided transmission links. The most relevant research works for this paper are those focused on MIMO arrays at the transmitter and the receiver, and include \cite{DBLP:journals/jsac/ZhangZ20a}, \cite{DBLP:journals/twc/PerovicTRF21}, \cite{DBLP:journals/wcl/AbrardoDRQ21}, \cite{DBLP:journals/tcom/HanZDZ22}, \cite{DBLP:journals/tcom/ZhouXLFN22}, \cite{DBLP:journals/corr/abs-2104-12108}. In \cite{DBLP:journals/jsac/ZhangZ20a}, the authors propose an alternating optimization algorithm to find a locally optimal solution by iteratively optimizing the transmit covariance matrix and the RIS reflection coefficients. In \cite{DBLP:journals/twc/PerovicTRF21}, the authors tackle a similar problem by proposing a more efficient numerical algorithm based on the projected gradient method. In  \cite{DBLP:journals/wcl/AbrardoDRQ21}, the authors consider an electromagnetic-consistent model in the presence of mutual coupling (non-diagonal RIS matrices) and generalize the weighted minimum mean square error algorithm for application to multi-RIS MIMO interference channels. In \cite{DBLP:journals/tcom/HanZDZ22}, the authors focus on (pure) line-of-sight channels, and generalize the approach in \cite{DBLP:journals/jsac/ZhangZ20a} for application to double-RIS deployments. The algorithm is specialized for application to far-field MIMO channels, hence hampering spatial multiplexing gains greater than the number of RIS, i.e., two. In \cite{DBLP:journals/tcom/ZhouXLFN22}, the authors introduce a numerical algorithm for jointly optimizing the MIMO transmit precoder, the matrix of reflection coefficients of the RIS, and the MIMO receive equalizer, by minimizing the data detection mean square error. In \cite{DBLP:journals/corr/abs-2104-12108}, the authors propose three numerical algorithms to maximize the sum-rate in RIS-aided MIMO broadcast channels. In all these papers, no closed-form expression for the matrix of reflection coefficients of the RIS is given, and all the algorithms are sub-optimal and are shown to be sensitive to the initialization point and to the setup of the parameters e.g., the step-size in gradient-based methods.

\subsection{Contributions}
Based on the analysis of the state of the art, we evince that there is no contribution that characterizes the optimal design of RIS-aided MIMO channels in terms of maximizing the mutual information. Also, there is no contribution that provides a closed-form, even approximated, expression for the reflection coefficients of RIS in line-of-sight and low-scattering channels. In this context, the present paper provides the following specific contributions.
\begin{itemize}
\item In a general RIS-aided MIMO system, we prove that the channel capacity is achieved by diagonalizing the end-to-end transmitter-RIS-receiver channel, and by applying the water-filling power allocation to the ordered product of the singular values of the (two individual) transmitter-RIS and RIS-receiver channels. The obtained capacity-achieving solution is shown to require an RIS with a non-diagonal matrix of reflection coefficients. Recently, non-diagonal RIS have been receiving attention from the research community, since they can provide better performance at the expense of an increase of the hardware and implementation complexity, e.g., \cite{DBLP:journals/twc/ShenCM22}, \cite{DBLP:journals/tvt/LiEHSMCH22}, \cite{TWC-BeyondDiagonal-2022}. To the best of the authors knowledge, however, no prior research work on non-diagonal RIS has proved the optimality of this architecture from an information-theoretic standpoint. Under the assumption of nearly-passive RIS, i.e., no power amplification is needed at the RIS, we prove that the water-filling power allocation is necessary only at the MIMO transmitter. We refer to this capacity-optimal RIS as a linear, nearly-passive, reconfigurable electromagnetic object (EMO). 

\item Even though an EMO is nearly-passive, i.e., it does not require power amplifiers, and is capacity-achieving, it is characterized by a matrix of reflection coefficients with a non-diagonal structure. In general, all the entries of the matrix may be non-zero. This entails a non-negligible implementation complexity to realize the associated configuration network \cite{DBLP:journals/twc/ShenCM22} or even the need of non-local designs \cite{doi:10.1126/sciadv.1602714}, \cite{DBLP:journals/pieee/RenzoDT22}. It is known, in addition, that RIS-aided systems characterized by non-diagonal matrices are more difficult to be optimized, in general \cite{DBLP:journals/wcl/AbrardoDRQ21}. Specifically, it is known that no closed-form solution for the optimal design of RIS is available in RIS-aided (pure) line-of-sight MIMO channels either \cite{DBLP:journals/tcom/HanZDZ22}. As detailed in \cite[Eq. (5)]{DBLP:journals/tcom/HanZDZ22}, this is due to the non-convexity of the optimization problem and the coupling among the optimization variables (the matrix of reflection coefficients of the RIS and the precoding matrix at the MIMO transmitter). It is known, in addition, that the main computational complexity is due to the presence of the RIS, and that the existing optimization algorithms are computational intensive, and are highly sensitive to the initialization point and the free parameters of the algorithms (e.g., the step size) \cite{DBLP:journals/corr/abs-2104-12108}. To overcome these difficulties, we introduce an approximated closed-form and low-complexity design for RIS, whose matrix of reflection coefficients is diagonal with unit-modulus entries. The reflection coefficients are given by the product of two focusing functions: one steering the RIS-aided signal towards the mid-point of the MIMO transmitter and one steering the RIS-aided signal towards the mid-point of the MIMO receiver. Notably, we prove that the proposed solution is exact in (pure) line-of-sight channels under the paraxial setup \cite{Miller00}. 

\item With the aid of extensive numerical simulations in line-of-sight (free-space) fading channels, we show that the proposed design for (diagonal) RIS offers performance (rate and degrees of freedom) close to that obtained by numerically solving non-convex optimization problems, which is known to entail a high computational complexity. Specifically, the numerical results unveil that the proposed solution works well even under non-paraxial setups, for which no analysis of the degrees of freedom  and no closed-form expression of the channel eigenfunctions are available in closed-form for RIS-aided links \cite{Miller00}, \cite{DBLP:journals/access/DecarliD21}. In addition, the numerical results show that the proposed approximated diagonal RIS offers performance very close to that of the EMO (with a non-diagonal structure), in several considered case studies. The proposed design can be considered as an approximate closed-form optimal solution for RIS-aided MIMO channels or a good-quality initial point for efficiently implementing currently-available non-convex optimization algorithms \cite{DBLP:journals/jsac/ZhangZ20a}, \cite{DBLP:journals/twc/PerovicTRF21}. Also, the proposed solution is applicable in every signal-to-noise-ratio regime, since the optimal precoding and decoding matrices at the MIMO transmit and MIMO receiver are obtained by using conventional singular-value decomposition methods \cite{DBLP:journals/ett/Telatar99}.
\end{itemize}

\subsection{Organization}
The rest of this paper is organized as follows. In Section II, we introduce the signal and system models. In Section III, we characterize the capacity-achieving structure of RIS in general MIMO channels. In Section IV, we introduce a sub-optimal design for diagonal RIS, which is proved to be optimal in line-of-sight channels under the paraxial setup. In Section V, we illustrate extensive simulation results over line-of-sight (free-space) channels, and non-paraxial system setups, to showcase the performance of the proposed closed-form sub-optimal design for RIS against state of the art numerical methods and the optimal EMO. In Section VI, conclusions are given.

\section{Signal and System Models}
\label{Sec:SysMod}

\begin{figure}[t!]
     \centering
         \includegraphics[width=\textwidth]{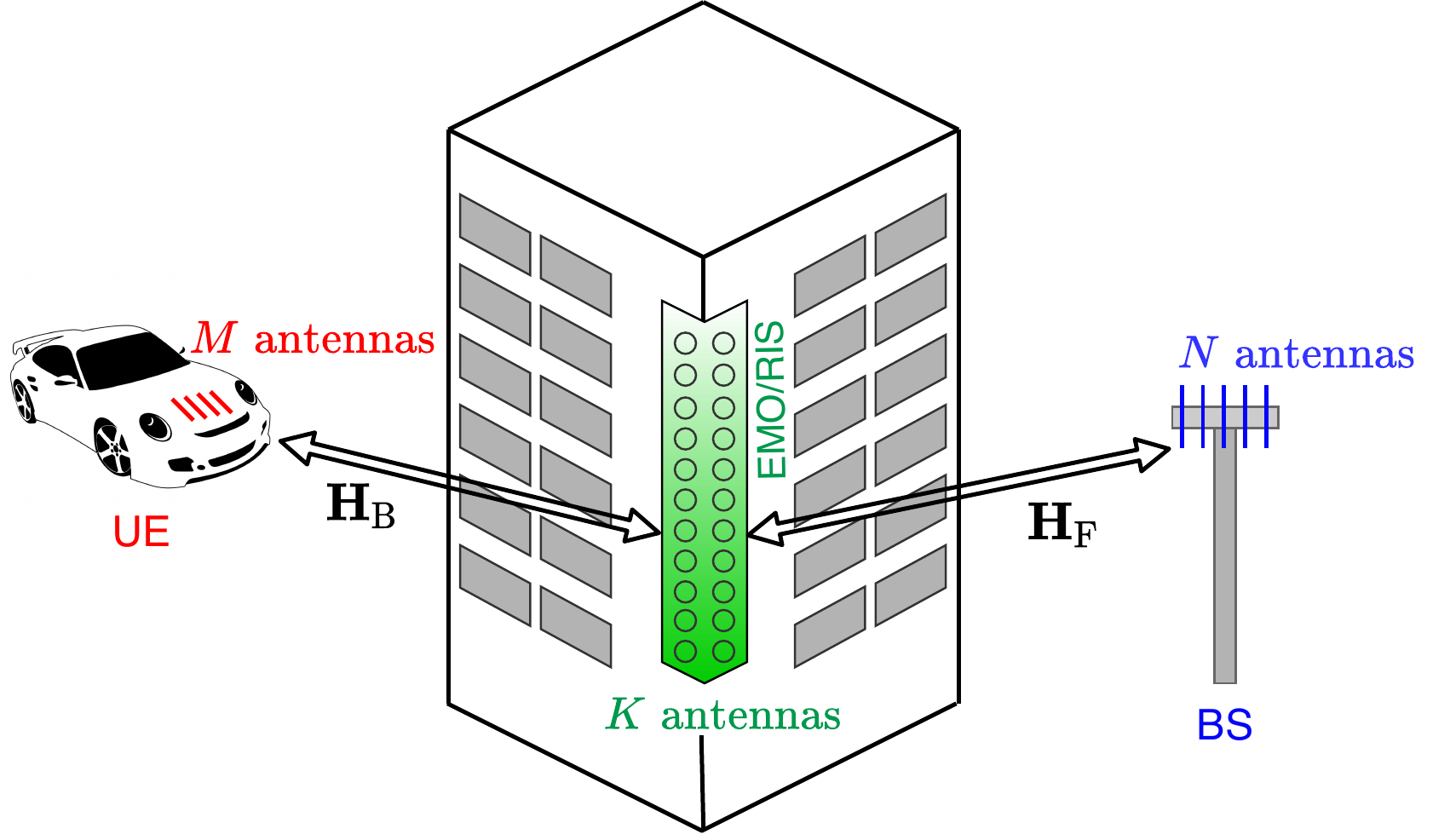}
         \caption{Considered scenario: A MIMO transmitter (a BS) communicates with a MIMO receiver (a UE) with the aid of an RIS.}
         \label{fig:Scenario}
\end{figure}

We consider the communication system sketched in Fig.~\ref{fig:Scenario}, where a MIMO transmitter, e.g., a base station (BS), communicates with a MIMO receiver, e.g., a user equipment (UE), with the aid of an RIS. The transmitter and the receiver are equipped with $\kt$ and $\kr$ antennas, respectively. The RIS is modeled as a MIMO array compromising $K$ reconfigurable scattering elements. As discussed in Section I, the antennas at the transmitter and receiver, and the reconfigurable elements at the RIS are assumed to be critically spaced at half-wavelength. 

The channel matrix between the transmitter and the RIS is denoted by $\Hf \in \mathbb{C}^{\kris \times \kt}$ (forward channel), and the channel matrix between the RIS and the receiver is denoted by $\Hb \in \mathbb{C}^{\kr\times \kris}$ (backward or scattered channel). No specific assumption about the characteristics of the channels are made, even though the focus of this paper is on analyzing the multiplexing gain of the considered RIS-aided MIMO system in line-of-sight and low-scattering channels, where no spatial multiplexing gain can be obtained in the far field \cite{tse_viswanath_2005}.

The end-to-end (transmitter-RIS-receiver) MIMO channel can be formulated as
\begin{equation}\label{eq:TXRXch_base}
    \Hphi=\Hb\PHI\Hf\
\end{equation}
where $\PHI$ denotes the $K \times K$ matrix of reflection coefficients of the RIS.

Typical designs for RIS assume that the matrix of reflection coefficients $\PHI$ is diagonal with unit-modulus entries \cite{TWC-BeyondDiagonal-2022}. This design eases the implementation complexity of RIS, but may not necessarily be optimal from the point of view of the power efficiency \cite{doi:10.1126/sciadv.1602714}, \cite{DBLP:journals/pieee/RenzoDT22}, performance \cite{DBLP:journals/twc/ShenCM22}, or may not consider the impact of mutual coupling in sub-wavelength (with inter-distances less than half-wavelength) implementations \cite{DBLP:journals/wcl/GradoniR21}. Collectively, as mentioned in Section I, these designs fall within the umbrella of non-diagonal RIS.  However, there exist no information-theoretic studies that quantify the optimality of non-diagonal RIS compared with diagonal RIS. 

To shed light on the optimal design of RIS from an information-theoretic standpoint, we depart by considering an arbitrary structure for the matrix $\PHI$, which is not necessarily diagonal. The only imposed constraint for $\PHI$ is to be a unitary matrix, i.e., $\PHI\PHI^{\herm} = {\bf{I}}_K$, where $(\cdot)^{\herm}$ is the Hermitian adjoint operator and ${\bf{I}}_K$ is the $K \times K$ identity matrix. The assumption of unitary matrix ensures that the RIS does not amplify the incident signals but relaxes the assumption of diagonal structure. This will be apparent in the next sections.

An RIS with a non-diagonal unitary matrix $\PHI$ is referred to as an EMO with the following characteristics: 
\begin{itemize}
    \item It is a noise-less device, as opposed to non-regenerative relays \cite{DBLP:journals/ojcs/RenzoNSDQLRPSZD20};
    \item It is a linear device whose output signals are linear combinations of the input signals;
    \item It is reconfigurable, i.e., the matrix $\PHI$ can be optimized as desired;
    \item It is nearly-passive, i.e., it does not amplify the incident signals thanks to the design constraint $\PHI\PHI^{\herm} = {\bf{I}}_K$. A formal proof for this statement is given in the next section.
\end{itemize}

To facilitate the analysis, the forward and backward channels can be formulated by invoking the singular-value decomposition (SVD) applied to $\Hf$ and $\Hb$. Specifically, let $\diag \{v_1, v_2, \ldots \}$ denote a general rectangular matrix whose main diagonal has elements given by $v_1$, $v_2$, $\cdots$, and let $\Sfi$ and $\Sbi$ denote the $\kt$ and $\kr$  (non-zero and zero) ordered singular values of $\Hf$ and $\Hb$, respectively. For simplicity, we assume $N < K$ and $M < K$. Also, let $N_{\text{F}}\le\kt$ and $N_{\text{B}}\le\kr$ be the numbers of non-zero singular values of $\Hf$ and $\Hb$, respectively. Then, we introduce the following matrices
\begin{IEEEeqnarray}{rCl}
    \Sf&=&\diag{\left\{\sigma_{\text{F}1}, \ldots, \sigma_{\text{F}N} \right\}} \, \in \mathbb{C}^{\kris\times\kt}  \\
    \Sb&=&\diag{\left\{\sigma_{\text{B}1}, \ldots, \sigma_{\text{B}M} \right\}} \, \in \mathbb{C}^{\kr\times \kris} 
\end{IEEEeqnarray}

In addition, let $\Uf \in \mathbb{C}^{\kris\times \kris}$, $\Vf \in \mathbb{C}^{\kt\times \kt}$, $\Ub \in \mathbb{C}^{\kr\times \kr}$, and $\Vb \in \mathbb{C}^{\kris\times \kris}$ be unitary matrices fulfilling the following equalities
\begin{IEEEeqnarray}{rCl}
    \Hf&=&\Uf\Sf\Vf^{\herm}  \label{eq:SVDf} \\
    \Hb&=&\Ub\Sb\Vb^{\herm} \, \label{eq:SVDb}
\end{IEEEeqnarray}

Therefore, the end-to-end channel $\HH$ can be expressed as
\begin{equation}\label{eq:TXRXch}
    \Hphi=\Hb\PHI\Hf=\Ub\Sb\Vb^{\herm}\PHI\Uf\Sf\Vf^{\herm}
\end{equation}

By definition, the maximum number of communication modes (i.e., the channel degrees of freedom) is ${\text{DoF}}_{\max}=\text{min}(\kt,\kr)$. The actual number of degrees of freedom highly depends on the properties of the channels $\Hf$ and $\Hb$, and on the matrix $\PHI$ of RIS reflection coefficients. The objective of this paper is to optimize the matrix $\PHI$ in order to maximize the end-to-end channel capacity given the channels $\Hf$ and $\Hb$.

\section{Capacity-Achieving Matrix of Reflection Coefficients}
\label{Sec:Optimum}
In this section, we identify the optimal structure for the unitary matrix $\PHI$ that is capacity-achieving. The received signal $\B{y}\in\mathbb{C}^{\kr}$ is
\begin{equation}
    \B{y} = \Hphi\B{x}+\B{n}
\end{equation}
where $\B{x}\in\mathbb{C}^{\kt}$ is the transmitted signal, $\B{n}\in\mathbb{C}^\kr$ is the additive white Gaussian noise (AWGN) with $\B{n} \sim {\mathcal{CN}}\left (\mathbf{0}, \sigma_{n}^2 \boldI_M \right )$, where $\sigma_{n}^2$ is the noise power and ${\mathcal{CN}}\left (\cdot, \cdot \right)$ denotes a multivariate complex Normal distribution. 

Also, let $\B{Q}$ be the positive semi-definite (i.e., ${\bf{Q}} \succeq 0$) covariance matrix of the transmitted signal, i.e., $\mathbb{E}[\B{x}\B{x}^{\herm}]=\B{Q}$, where $\mathbb{E}[\cdot]$ is the expectation operator. Then, the  transmit power constraint is $\operatorname{tr}(\B{Q})\leq P$, where $\operatorname{tr}(\cdot)$ is the trace operator and $P$ is the total transmit power.

According to \cite{tse_viswanath_2005}, the channel capacity is obtained by maximizing the mutual information between the transmitted and received symbols as a function of ${\bf{Q}}$ and $\bf{\Phi }$
\begin{IEEEeqnarray}{rCl} \label{eq:MutualInf}
C= \mathop {\max }\limits_{\scriptstyle{\bf{Q}} \succeq 0,{\rm{tr}}\left( {\bf{Q}} \right) \le P\atop
\scriptstyle{\bf{\Phi }}{{\bf{\Phi }}^\dag } = {{\bf{I}}_K}} I\left( {{\bf{x}},{\bf{y}}} \right) = \log \det \left( {{\bf{I}}_M + {\bf{HQ}}{{\bf{H}}^\dag }} \right)
\end{IEEEeqnarray}
where $\det \left( \cdot \right)$ is the determinant of a square matrix.

The channel capacity in \eqref{eq:MutualInf} and how to achieve it are given in the following proposition.
\begin{proposition}
The end-to-end channel capacity in \eqref{eq:MutualInf} is 
\begin{equation}
C = \sum\limits_{i = 1}^{\min \left\{ {N,M} \right\}} {\log \left( {1 + \sigma _{{\rm{F}}i}^2\sigma _{{\rm{B}}i}^2\frac{{{P_i}}}{{\sigma _n^2}}} \right)} \label{Eq_C}
\end{equation}
where ${P_i} = {\left( {\mu  - \frac{{\sigma _n^2}}{{\sigma _{{\rm{F}}i}^2\sigma _{{\rm{B}}i}^2}}} \right)^ + } = \max \left( {0,\mu  - \frac{{\sigma _n^2}}{{\sigma _{{\rm{F}}i}^2\sigma _{{\rm{B}}i}^2}}} \right)$ and $\mu$ is obtained by fulfilling the identity $\sum\nolimits_{i = 1}^{\min \left\{ {N,M} \right\}} {{P_i}}  = P$. In addition, $C$ is attained by setting
\begin{IEEEeqnarray}{l}
{\bf{Q}} = {\rm{diag}}\left\{ {{P_1},{P_2}, \ldots ,} \right\}\\
{\bf{\Phi }} = {{\bf{V}}_{\rm B}}{\bf{U}}_{\rm F}^\dag \label{Eq_EMOopt}
\end{IEEEeqnarray}
\end{proposition}
\begin{proof}
See Appendix A.
\end{proof}

Based on Proposition 1, the following comments are in order.
\begin{itemize}

\item The capacity in \eqref{Eq_C} generalizes the capacity for MIMO channels in \cite{tse_viswanath_2005}, by considering the impact of an EMO (the optimal non-diagonal unitary RIS). Similar to \cite{tse_viswanath_2005}, the end-to-end channel capacity is attained by diagonalizing the end-to-end transmitter-RIS-receiver channel, and by applying the water-filling power allocation to the ordered product of the singular values of the transmitter-RIS and RIS-receiver channels taken individually. 

\item The capacity in \eqref{Eq_C} holds true for any channel models, and offers a simple closed-form solution for the optimal covariance matrix at the transmitter and the optimal design for the RIS. This is in stark contrast with computational intensive algorithms that are typically utilized for optimizing diagonal and non-diagonal RIS, e.g., \cite{DBLP:journals/tcom/HanZDZ22}, \cite{DBLP:journals/corr/abs-2104-12108}, \cite{TWC-BeyondDiagonal-2022}.

\item An RIS with a diagonal matrix ${\bf{\Phi }}$ of unit-modulus reflection coefficients is, in general, not capacity-achieving. An RIS with a non-diagonal unitary matrix ${\bf{\Phi }}$ may be, on the other hand, capacity-achieving. Also, the optimal design is available in a closed-form expression based on the SVD decomposition of the transmitter-RIS and RIS-receiver channels.

\item By direct inspection of \eqref{Eq_C}, it is apparent that considering a non-diagonal but unitary matrix ${\bf{\Phi }}$ is sufficient for ensuring that the RIS does not result in any power amplification, thus making the EMO a nearly-passive device like a diagonal RIS.

\item Similar to \cite{tse_viswanath_2005}, the capacity in \eqref{Eq_C} holds true for any operating signal-to-noise ratio.

\item The formula in \eqref{Eq_C} resembles the capacity of non-regenerative relays \cite{SAM-2006} and \cite{DBLP:journals/twc/TangH07}. The main differences between the EMO and a non-regenerative relay are that the former is noise-free and no power allocation at the EMO is needed, since it is characterized by a unitary matrix. This makes the proof of Proposition 1 slightly different from \cite{SAM-2006} and \cite{DBLP:journals/twc/TangH07}, since there is no need to apply a whitening matrix. Due to the capability of non-regenerative relays to amplify the received signal, the optimal power allocation strategy is known, in a closed-form expression, only if the transmitter applies an equal power allocation strategy \cite{DBLP:journals/twc/TangH07}. Otherwise, the optimal power allocation needs to be computed by solving an optimization algorithm. The optimal power allocation at the MIMO transmitter and the optimal matrix of the RIS reflection coefficients are, on the other hand, formulated in closed-form expressions in Proposition 1.

\end{itemize}

Proposition 1 unveils that the optimal design for an RIS with a unitary matrix of reflection coefficients is non-diagonal. Non-diagonal RIS have attracted recent interest in the literature, and, supported by Proposition 1, they usually outperform diagonal RIS \cite{TWC-BeyondDiagonal-2022}. However, their practical implementation is more difficult, especially due to the complex configuration network to realize non-diagonal matrices. Based on the optimal design for unitary RIS in \eqref{Eq_EMOopt}, it is, however, possible to investigate approximated diagonal (hence unitary by definition) designs for RIS that provide a rate close to the end-to-end channel capacity in \eqref{Eq_C}. This is discussed next.

\section{Proposed Design for Nearly-Optimal Diagonal RIS}
\label{Sec:RIS}

In this section, we propose an approximated design for RIS whose matrix of reflection coefficients is diagonal, and analyze the conditions under which it is either optimal or nearly-optimal. As a criterion of optimality, we consider the end-to-end channel capacity in Proposition 1. In other words, a diagonal RIS is deemed optimal if the resulting end-to-end transmitter-RIS-receiver channel can be diagonalized and the capacity in \eqref{Eq_C} is attained. For the avoidance of doubt, the diagonal matrix of reflection coefficients of the RIS is denoted by $\tilde{\PHI}$.

Without loss of generality, we assume that the $K \times K$ diagonal matrix $\tilde{\PHI}$ is given by the product of two $K \times K$ diagonal matrices $\mathbf{F}_{\text{RB}}$ and $\mathbf{F}_{\text{RF}}$, as follows
\begin{equation}
\tilde{\PHI}=\mathbf{F}_{\text{RB}} \mathbf{F}_{\text{RF}}^\dag  \label{eq:phi_focus}
\end{equation}

The definition and meaning of $\mathbf{F}_{\text{RB}}$ and $\mathbf{F}_{\text{RF}}$ based on the considered criterion of optimality are given next. Furthermore, let us express the unitary matrices $\Uf \in \mathbb{C}^{\kris\times \kris}$, $\Vf \in \mathbb{C}^{\kt\times \kt}$, $\Ub \in \mathbb{C}^{\kr\times \kr}$, and $\Vb \in \mathbb{C}^{\kris\times \kris}$ in \eqref{eq:SVDf} and \eqref{eq:SVDb}, as follows
\begin{IEEEeqnarray}{rCl}
    \Vf&=&\Ftx\Ptx \in \mathbb{C}^{\kt\times \kt} \label{eq:Vf} \\
    \Uf&=&\Frr\Prr \in \mathbb{C}^{\kris\times \kris} \label{eq:Uf}\\
    \Vb&=&\Frt\Prt \in \mathbb{C}^{\kris\times \kris} \label{eq:Vb} \\
    \Ub&=&\Frx\Prx \in \mathbb{C}^{\kr\times \kr} \label{eq:Ub}
\end{IEEEeqnarray}
where $\Ftx$ and $\Ftx$ are $N \times N$ and $M \times M$ diagonal matrices, respectively, $\Ptx$ and $\Prx$ are $N \times N$ and $M \times M$ non-diagonal unitary matrices, respectively, and $\Prr$ and $\Prt$ are $K \times K$ non-diagonal unitary matrices. The analytical formulation in \eqref{eq:Vf}-\eqref{eq:Ub} holds, without loss of generality, for any choice of the diagonal matrices $\mathbf{F}_{\text{RB}}$ and $\mathbf{F}_{\text{RF}}$ in \eqref{eq:phi_focus}, and for any channels $\Hf$ and $\Hb$, by appropriately choosing the non-diagonal unitary matrices $\Ptx$, $\Prx$, $\Prr$, and $\Prt$.

Based on these definitions, the end-to-end channel in \eqref{eq:TXRXch} can be reformulated as
\begin{equation}
    \HH=\Hb \tilde \PHI \Hf=\left(\Frx\Prx \Sb \Prt^{\herm}\mathbf{F}_{\text{RB}}^{\herm}\right) \left(\mathbf{F}_{\text{RB}} \mathbf{F}_{\text{RF}}^\dag \right) \left(\mathbf{F}_{\text{RF}} \Prr \Sf \Ptx^{\herm}\Ftx^{\herm}\right) \label{eq:H}
\end{equation}

Since $\mathbf{F}_{\text{RB}}^{\herm} \mathbf{F}_{\text{RB}} = {\bf{I}}_K$ and $\mathbf{F}_{\text{RF}}^{\herm} \mathbf{F}_{\text{RF}} = {\bf{I}}_K$, we obtain
\begin{IEEEeqnarray}{rCl}
    \HH &=& \left(\Frx\Prx \Sb\right) \left(\Prt^{\herm} \Prr\right) \left(\Sf \Ptx^{\herm}\Ftx^{\herm}\right)
\end{IEEEeqnarray}

Since, by definition, $\Ub = \Frx\Prx$ and $\Vf^{\herm} = \Ptx^{\herm}\Ftx^{\herm}$, we obtain
\begin{IEEEeqnarray}{rCl}
    \HH &=& \left(\Ub \Sb\right) \left(\Prt^{\herm} \Prr\right) \left(\Sf \Vf^{\herm}\right) \label{Eq_Ch_Approx}
\end{IEEEeqnarray}

From Proposition 1, the end-to-end channel capacity in \eqref{Eq_C} is obtained by diagonalizing the transmitter-RIS-receiver channel, as follows
\begin{IEEEeqnarray}{rCl}
   \HH_{\rm{diag}} &=& \Hb \left({{\bf{V}}_{\rm B}}{\bf{U}}_{\rm F}^\dag\right) \Hf = \Ub \Sb \Sf \Vf^{\herm} = \Ub {\mathbf\Gamma} \Vf^{\herm} \label{Eq_Ch_Diag}
\end{IEEEeqnarray}
where ${\mathbf\Gamma} = \Sb \Sf \in \mathbb{C}^{\kr \times \kt}$ is an $\kr \times \kt$ diagonal matrix.

By comparing \eqref{Eq_Ch_Approx} and \eqref{Eq_Ch_Diag}, we evince that a diagonal RIS is optimal from the end-to-end channel capacity standpoint if $\Prt^{\herm} \Prr= {\bf{I}}_K$, since the end-to-end channel is diagonalized according to Proposition 1, i.e., $\HH  = \HH_{\rm{diag}}$. In other words, a diagonal RIS would be nearly-optimal, from the end-to-end channel capacity standpoint, in all setups in which $\HH  \approx \HH_{\rm{diag}}$. This provides a simple criterion for optimizing the matrix $\PHI$, while ensuring a low-complexity implementation (i.e., a diagonal structure) and near-optimality based on Proposition 1. Therefore, several designs for RIS may be conceived, offering the desired tradeoff between achievable rate and implementation complexity, in contrast to fully-connected non-diagonal RIS \cite{TWC-BeyondDiagonal-2022}.

Based on the design criterion $\HH  \approx \HH_{\rm{diag}}$ for optimizing the diagonal matrix $\tilde \PHI$, it is worth analyzing whether there exist channel models and setups under which a diagonal RIS is strictly optimal, i.e., $\HH  = \HH_{\rm{diag}}$. In addition, it is worth analyzing whether approximated closed-form expressions for $\tilde \PHI$ that fulfill the condition $\HH  \approx \HH_{\rm{diag}}$ can be identified, especially in line-of-sight or low-scattering channels. These aspects are discussed in the next two sub-sections.

\begin{figure}[!t]
     \centering
         \includegraphics[width=\textwidth]{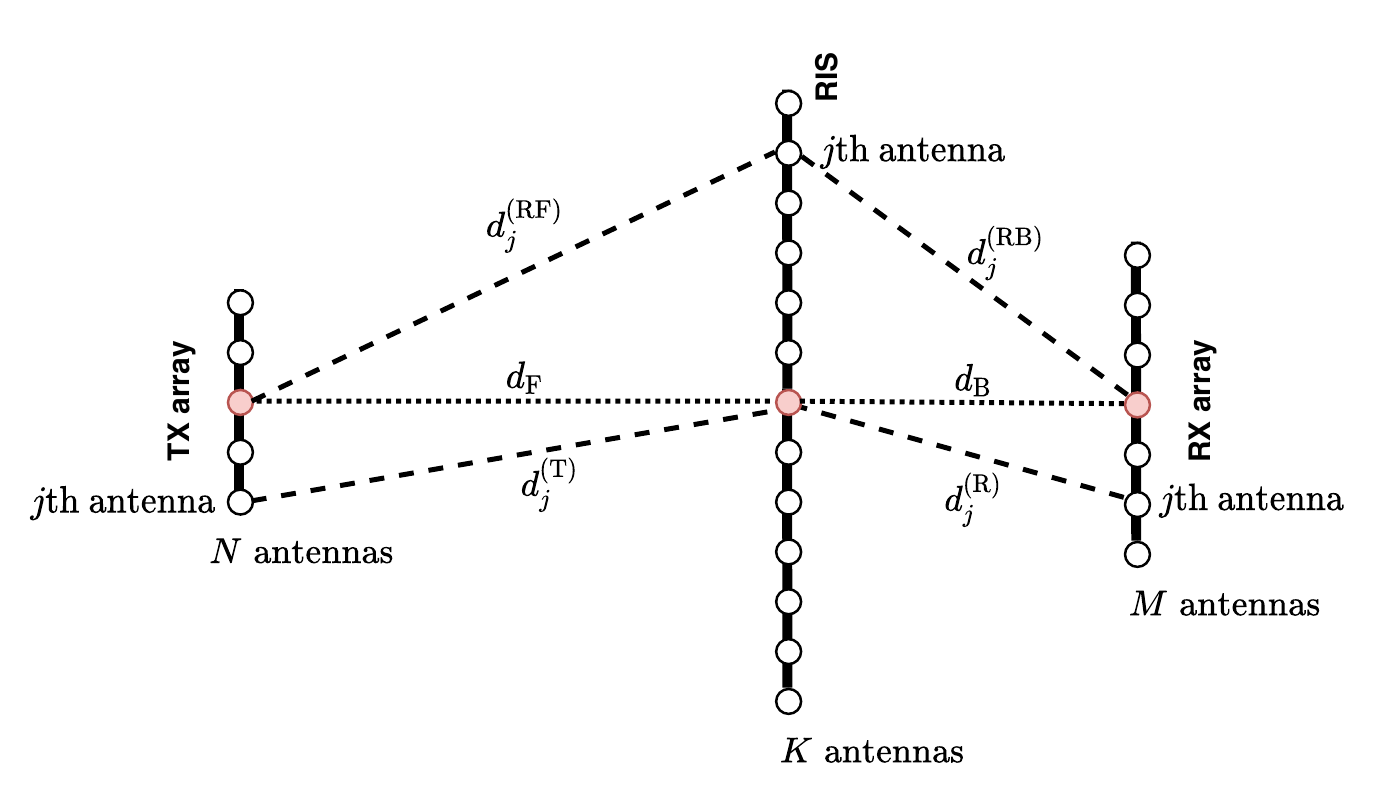}
         \caption{Example of deployment scenario in which a diagonal RIS is optimal from the capacity standpoint: Line-of-sight channels and paraxial setup.}
         \label{fig:Scheme}
\end{figure}
\subsection{Line-of-Sight Channels and Paraxial Setup}
A known and important case study in which a diagonal RIS is strictly optimal, i.e., $\HH  = \HH_{\rm{diag}}$, is sketched in Fig. \ref{fig:Scheme} for, without loss of generality, linear arrays \cite{Anna-Thaning-2003}. Specifically, the scenario in Fig. \ref{fig:Scheme} is representative of a refracting RIS, enabling, e.g., outdoor-to-indoor or room-to-room communications \cite{DBLP:journals/cm/ZhangZDTRDHPS22}. We consider (pure) line-of-sight channels where the MIMO transmitter, the RIS, and the MIMO receiver are in the near field of each other (Fresnel region), the mid-points of the three antenna arrays are perfectly aligned, and the MIMO transmitter and MIMO receiver are located on the opposite sides of the RIS. This network topology is usually referred to as paraxial setup \cite{Miller00}. Also, we assume that the distance between the mid-points of the MIMO transmitter and the RIS ($d_{\rm{F}}$ in Fig. \ref{fig:Scheme}), and the distance between the mid-points of the RIS and the MIMO receiver ($d_{\rm{B}}$ in Fig. \ref{fig:Scheme}) are the same, i.e., $d_{\rm{F}}=d_{\rm{B}}$, and $N=M$. In detail, the line-of-sight transmitter-RIS and RIS-receiver links are denoted by ${\bf H}_{\rm F}=\left\{h_{nk}^{(\rm{F})}\right\}$ and ${\bf H}_{\rm B}=\left\{h_{km}^{(\rm{B})}\right\}$
\begin{align}
    h_{nk}^{(\rm{F})}&=\frac{1}{4\pi d_{nk}^{(\rm{F})}}e^{\jmath \kappa d_{nk}^{(\rm{F})}} \label{eq:FChannel}\\
    h_{km}^{(\rm{B})}&=\frac{1}{4\pi d_{km}^{(\rm{B})}}e^{\jmath \kappa d_{km}^{(\rm{B})}} \label{eq:BChannel}
\end{align}
where $d_{nk}^{(\rm{F})}$ is the distance between the $n$th antenna of the MIMO transmitter and the $k$th element of the RIS, $d_{km}^{(\rm{B})}$ is the distance between the $k$th element of the RIS and the $m$th antenna of the MIMO receiver, $\jmath$ is the imaginary unit,  and $\kappa = 2 \pi / \lambda$ with $\lambda$ being the wavelength.

Under these conditions, it is known from \cite{Miller00} (for surfaces or planar antenna arrays) and \cite{Anna-Thaning-2003} (for lines or linear arrays) that $\Prt^{\herm} \Prr= {\bf{I}}_K$, and, notably, the matrices $\Prt$ and $\Prr$ are obtained by spatially sampling, at half-wavelength, prolate spheroidal functions \cite{SlePol:J61a}. In addition, the diagonal matrix $\tilde{\PHI}=\mathbf{F}_{\text{RB}} \mathbf{F}_{\text{RF}}^\dag $ is available in a closed-form expression, as
\begin{IEEEeqnarray}{l}
    \Frr=\diag{\left\{e^{\jmath\kappa {d_{1}^{(\text{RF})}}}, \ldots, e^{\jmath\kappa {d_\kris^{(\text{RF})}}}\right\}}\,\in \mathbb{C}^{\kris\times \kris} \label{Eq_FocusingRF} \\
    \Frt=\diag{\left\{e^{\jmath\kappa {d_{1}^{(\text{RB})}}}, \ldots, e^{\jmath\kappa {d_{\kris}^{(\text{RB})}}}\right\}}\,\in \mathbb{C}^{\kris\times \kris} \label{Eq_FocusingRB} 
\end{IEEEeqnarray}
where, as illustrated in Fig. \ref{fig:Scheme}, $d_j^{(\text{RF})}$ is the distance between the mid-point of the MIMO transmitter and the $j$th element of the RIS, and $d_j^{(\text{RB})}$ is the distance between the $j$th element of the RIS and the mid-point of the MIMO receiver. In \cite{Miller00}, the functions in \eqref{Eq_FocusingRF} and \eqref{Eq_FocusingRB} are referred to as focusing functions. In this scenario, in addition, the functions $\Ftx$, $\Ftx$, $\Ptx$, and $\Prx$ in \eqref{eq:Vf} and \eqref{eq:Ub} can be computed analogously, as elaborated in \cite{Miller00}.

\subsection{Approximated Closed-Form Expression for $\tilde{\PHI}=\mathbf{F}_{\rm{RB}} \mathbf{F}_{\rm{RF}}^\dag $}
The case study analyzed in the previous section is a very special setup. However, it is instrumental to motivate the proposed approximated design for $\tilde{\PHI}=\mathbf{F}_{\rm{RB}} \mathbf{F}_{\rm{RF}}^\dag$. Based on recent results, e.g., \cite{DBLP:journals/tcom/AbrardoDR21}, \cite{DBLP:journals/jstsp/PanZZHWPRRSZZ22}, it is considered difficult to optimize an RIS based on instantaneous channel state information, e.g., based on estimates of $\Hf$ and $\Hb$ that account for the impact of small-scale fading. This is due to the associated large pilot overhead, which typically grows with the product of the number of RIS elements ($K$) and the number of users in the network. A more pragmatic choice is to optimize the RIS based on long-term channel statistics, e.g., the locations of the transmitter, RIS, and receiver, ignoring the small-scale fading. On the other hand, the MIMO transmitter and MIMO receiver are optimized based on full channel state information, similar to legacy systems, once the RIS is optimized. The focusing functions in \eqref{Eq_FocusingRF} and \eqref{Eq_FocusingRB} are optimal in line-of-sight channels and the RIS is optimized based only on the transmission distances. Under non-paraxial setups, the identity $\Prt^{\herm} \Prr= {\bf{I}}_K$ is, however, never exactly fulfilled, since $\Prt \ne \Prr$ in general. However, for any arbitrary matrix $\tilde{\PHI}$, the numerical computation of $\Prr$ and $\Prt$ is known to be straightforward \cite{DBLP:journals/tcom/HanZDZ22}, since, as apparent from \eqref{eq:Uf} and \eqref{eq:Vb}, it can be obtained by simply computing the SVD of the channels $\Hf$ and $\Hb$, respectively. The numerical computation of $\Prr$ and $\Prt$ can, in addition, partially compensate for any sub-optimality arising from the approximated choice for $\mathbf{F}_{\rm{RF}}$ and $\mathbf{F}_{\rm{RB}}$ in non-paraxial setups.

Based on these considerations, we propose the following approximated design for $\tilde{\PHI}$.
\begin{itemize}
\item Regardless of the channel model and the network topology, the diagonal matrix of reflection coefficients of the RIS is set to $\tilde{\PHI}=\mathbf{F}_{\rm{RB}} \mathbf{F}_{\rm{RF}}^\dag $, with $\mathbf{F}_{\rm{RF}}$ and $\mathbf{F}_{\rm{RB}}$ given in  \eqref{Eq_FocusingRF} and \eqref{Eq_FocusingRB}, respectively. 

\item Given $\tilde{\PHI}$, the precoding and decoding matrices at the MIMO transmitter and MIMO receiver, respectively, are obtained by applying the SVD to the end-to-end transmitter-RIS-receiver channel $\HH=\Hb \tilde \PHI \Hf = \Hb \mathbf{F}_{\rm{RB}} \mathbf{F}_{\rm{RF}}^\dag \Hf$ \cite{DBLP:journals/ett/Telatar99}.
\end{itemize}

We expect that the proposed approximated design for $\tilde{\PHI}$ is sufficiently accurate in line-of-sight and low-scattering channels, e.g., for application in the millimeter-wave, sub-terahertz, and terahertz frequency bands, and the location of the RIS is optimized in order to maximize the end-to-end performance given the locations of the MIMO transmitter and MIMO receiver. The optimality, from the end-to-end channel capacity standpoint, of the proposed approximated diagonal design is analyzed, with the aid of extensive numerical simulations, in the next section. Specifically, the rate obtained with the proposed design is compared against currently available numerical methods that jointly optimize the transmitter, the diagonal RIS, and the receiver \cite{DBLP:journals/tcom/AbrardoDR21}, as well as the EMO (non-diagonal RIS) based on Proposition 1.

\begin{figure}[t!]
     \centering
         \includegraphics[width=\columnwidth]{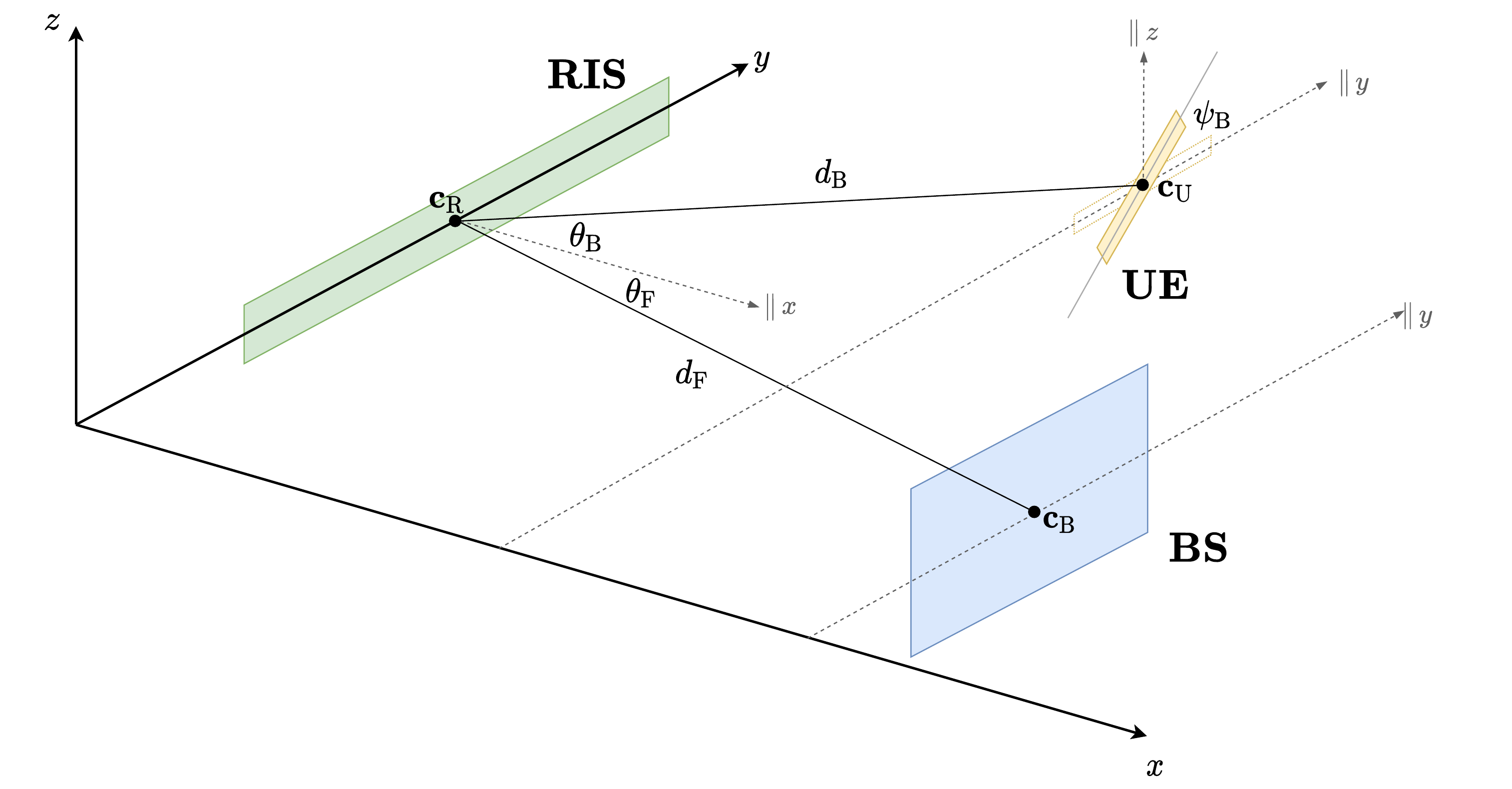}
         \caption{Considered network topology.}
         \label{fig:scenario_genaral}
\end{figure}
\section{Numerical Results}
\label{Sec:Results}

In this section, we analyze the effectiveness of the proposed closed-form expression for diagonal RIS with the aid of numerical simulations. Specifically, we compare the rate obtained by three schemes.
\begin{itemize}

\item The EMO in Proposition 1, which is the optimal capacity-achieving non-diagonal design for RIS that is obtained as the solution of \eqref{eq:MutualInf}. The capacity of this scheme is given by \eqref{Eq_C} without the need of solving any non-convex optimization problems. Only the SVD of the two channels $\Hf$ and $\Hb$ needs to be computed. This scheme is denoted by ``Non-diagonal RIS'' (ND-RIS) in the figures.

\item The conventional joint optimization of the covariance matrix of the transmitted symbols and the diagonal matrix of the RIS, by numerically solving the following non-convex optimization problem
\begin{IEEEeqnarray}{rCl} 
\mathop {\max }\limits_{\scriptstyle{\bf{Q}} \succ 0,{\rm{tr}}\left( {\bf{Q}} \right) \le P\atop
\scriptstyle\left| {{\bf{\Phi }}\left( {k,k} \right)} \right| = 1,\angle {\bf{\Phi }}\left( {k,k} \right) \in \left[ {0,2\pi } \right),{\bf{\Phi }}\left( {k,h \ne k} \right) = 0} \log \det \left( {\bf{I}}_M + \left({\bf{\Hb\PHI\Hf}}\right){\bf{Q}}\left({\bf{\Hb\PHI\Hf}} \right)^\dag\right)
\end{IEEEeqnarray}
where $\left| a \right|$ and $\angle a $ denote the absolute value and the phase of $a$, respectively. In this case, the algorithm in \cite{DBLP:journals/tcom/AbrardoDR21} is utilized. This scheme is denoted by ``Diagonal RIS - Numerical'' (D-RIS-NUM) in the figures.

\item The proposed approximated design for RIS, according to which the diagonal matrix of the RIS reflection coefficients is $\PHI= \tilde{\PHI}=\mathbf{F}_{\rm{RB}} \mathbf{F}_{\rm{RF}}^\dag $ with $\mathbf{F}_{\rm{RF}}$ and $\mathbf{F}_{\rm{RB}}$ given in \eqref{Eq_FocusingRF} and \eqref{Eq_FocusingRB}, respectively, and the covariance matrix of the transmitted symbols is the solution of the following convex optimization problem
\begin{IEEEeqnarray}{rCl} 
\mathop {\max }\limits_{\scriptstyle{\bf{Q}} \succ 0,{\rm{tr}}\left( {\bf{Q}} \right) \le P\atop} \log \det \left( {\bf{I}}_M + \left({\bf{\Hb \left(\mathbf{F}_{\rm{RB}} \mathbf{F}_{\rm{RF}}^\dag\right) \Hf}}\right){\bf{Q}}\left({\bf{\Hb \left(\mathbf{F}_{\rm{RB}} \mathbf{F}_{\rm{RF}}^\dag\right) \Hf}} \right)^\dag\right)
\end{IEEEeqnarray}
whose solution is known to be the water-filling power allocation applied to the end-to-end channel, considering $\PHI= \tilde{\PHI}=\mathbf{F}_{\rm{RB}} \mathbf{F}_{\rm{RF}}^\dag $ fixed \cite{DBLP:journals/ett/Telatar99}. This scheme is denoted by ``Diagonal RIS - Proposed'' (D-RIS-FOC) in the figures. In particular, the focusing functions in \eqref{Eq_FocusingRF} and \eqref{Eq_FocusingRB} assume, unless stated otherwise, that the focusing points are the mid-points of the MIMO transmitter and MIMO receiver, respectively. To analyze the sensitivity of the proposed solution, we analyze the impact of considering focusing points different from the mid-points of the MIMO transmitter and MIMO receiver as well. For generality, we use the notation $\PHI= \tilde{\PHI}=\mathbf{F}_{\rm{RB}}\left(\bf{p}_{\rm{RB}}\right) \mathbf{F}_{\rm{RF}}\left(\bf{p}_{\rm{RF}}\right)^\dag $, where $\bf{p}_{\rm{RF}}$ is the focusing point at the MIMO transmitter and $\bf{p}_{\rm{RB}}$ is the focusing point the MIMO receiver.
\end{itemize}

The considered network topology is illustrated in Fig. \ref{fig:scenario_genaral}. The mid-points of the MIMO transmitter, RIS, and MIMO receiver are denoted by $\mathbf{c}_{\rm B}$, $\mathbf{c}_{\rm R}$, $\mathbf{c}_{\rm U}$, respectively, and they lay on the $z=0$ plane. The distances between the mid-points of the MIMO transmitter and the RIS, and the mid-points of the RIS and the MIMO receiver are denoted by $d_{\rm{F}}$ and $d_{\rm{B}}$, respectively. The angles between the lines connecting the mid-points of the MIMO transmitter and the RIS, and the mid-points of the RIS and the MIMO receiver and the normal to the RIS are denoted by $\theta_{\rm F}$ and $\theta_{\rm B}$, respectively. Also, the angles between the MIMO transmitter and the $z$-axis and between the MIMO receiver and the $z$-axis are denoted by $\psi_{\rm F}$ and $\psi_{\rm B}$, respectively. The RIS is always kept parallel to the $yz$ plane. Therefore, the MIMO transmitter, RIS, and MIMO receiver are parallel to each other if $\psi_{\rm F}=\psi_{\rm B}=0$.

To validate the proposed design, the line-of-sight (free-space) channel model in \eqref{eq:FChannel} and \eqref{eq:BChannel} is utilized. Unless stated otherwise, the numerical results are obtained by assuming a carrier frequency of 28 GHz. The MIMO transmitter comprises $N=32$ antennas (8 along the $y$-axis and 4 along the $z$ axis), the RIS comprises $K=1024$ elements (512 along the $y$-axis and 2 along the $z$ axis), and the MIMO receiver comprises $M=16$ antennas (16 along the $y$-axis and 1 along the $z$ axis). The noise variance is set to $\sigma_n^2=-97~\deci\bel\milli$ and the total transmit power is set to $P=0~\deci\bel\milli$. The inter-distance of the antenna elements at the MIMO transmitter, MIMO receiver, and at the RIS is equal to half-wavelength.

\begin{figure}[t!]
     \centering
         \includegraphics[width=\columnwidth]{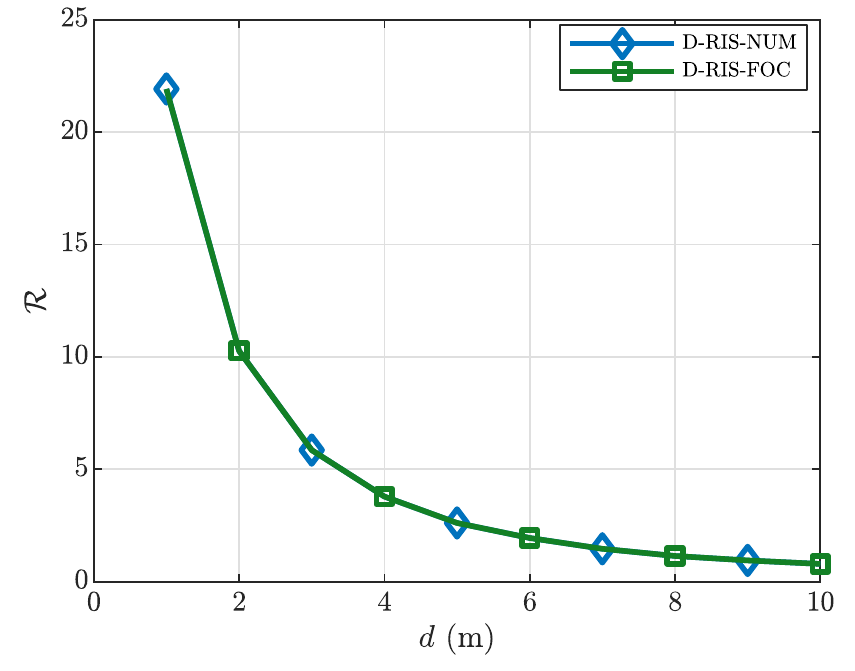}
         \caption{Achievable rate in the paraxial deployment -- Comparison between the ``Diagonal RIS - Numerical'' (D-RIS-NUN) and ``Diagonal RIS - Proposed'' (D-RIS-FOC) schemes. Setup: $d_{\rm{F}}=d_{\rm{B}}=d$, $\theta_{\rm F}=0$, $\theta_{\rm B}=\pi$, and $\psi_{\rm F}=\psi_{\rm B}=0$.}
         \label{fig:paraxial_performance}
\end{figure}
\begin{figure}[t!]
     \centering
         \includegraphics[width=\columnwidth]{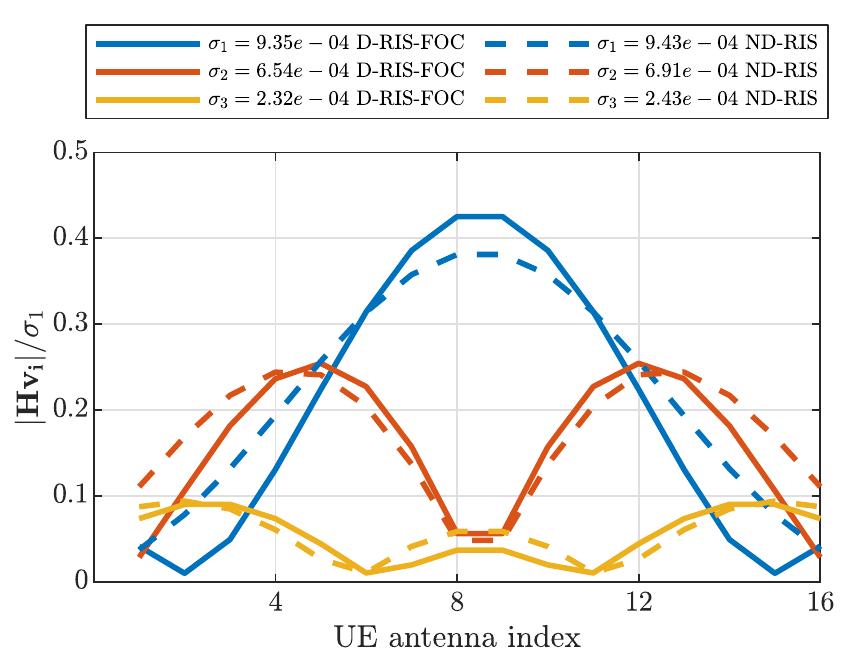}
         \caption{Achievable rate in the paraxial deployment -- Beam projection, normalized by the highest singular value ($\sigma_1=\sigma_{{\rm B}1}\sigma_{{\rm F}1}$), on the linear antenna array of the MIMO receiver: Comparison between the ``Non-diagonal RIS'' (ND-RIS) and ``Diagonal RIS - Proposed'' (D-RIS-FOC) schemes. Setup: $d_{\rm{F}}=d_{\rm{B}}=d = 7$ meter, $\theta_{\rm F}=0$, $\theta_{\rm B}=\pi$, and $\psi_{\rm F}=\psi_{\rm B}=0$.}
         \label{fig:paraxial_beams}
\end{figure}
\subsection{Paraxial Setup}
In this sub-section, we compare the achievable rate between the ``Diagonal RIS - Numerical'' (D-RIS-NUM) and ``Diagonal RIS - Proposed'' (D-RIS-FOC) schemes in the paraxial setup, similar to the deployment in Fig. \ref{fig:Scheme}. This is obtained by setting $d_{\rm{F}}=d_{\rm{B}}=d$, $\theta_{\rm F}=0$, $\theta_{\rm B}=\pi$, and $\psi_{\rm F}=\psi_{\rm B}=0$. As for the focusing centers, we set $\bf{p}_{\rm{RB}} = \mathbf{c}_{\rm B}$ and $\bf{p}_{\rm{RF}} = \mathbf{c}_{\rm U}$, according to Section IV-B. 

The results in Fig. \ref{fig:paraxial_performance} confirm the effectiveness of the proposed approach, according to the theoretical analysis of this case study. In Fig. \ref{fig:paraxial_beams}, we consider a similar setup with $d=7$ meters, and compare the ``Non-diagonal RIS'' (ND-RIS) and ``Diagonal RIS - Proposed'' (D-RIS-FOC) schemes in terms of projected beams. For both schemes, specifically, we compute the matrix $\Vf \in \mathbb{C}^{\kt\times \kt}$ from the corresponding SVD, and we then compute the projections $\|\mathbf{H}\mathbf{v}_{{\rm F}i}\|$ with $\mathbf{v}_{{\rm F}i}$ being the $i$th vector in $\Vf$. This projection represents how the energy of the transmitted signal is distributed across the MIMO receiver. The profiles of the obtained beams closely match with each other. The minor differences are due to the fact that the transmission distances of all the antenna elements are not exactly the same, and the Fresnel approximation applied in \cite{Miller00} is less accurate in the considered near field setup. Therefore, the focusing functions in \eqref{Eq_FocusingRF} and \eqref{Eq_FocusingRB} constitute an approximation.

\begin{figure*}[!t]
	\centering
	\subfloat[width=\columnwidth][Mid-point of the MIMO transmitter $\mathbf{c}_{\rm B}=(2.0,13.0,0.0)^T$]{\includegraphics[width=\columnwidth]{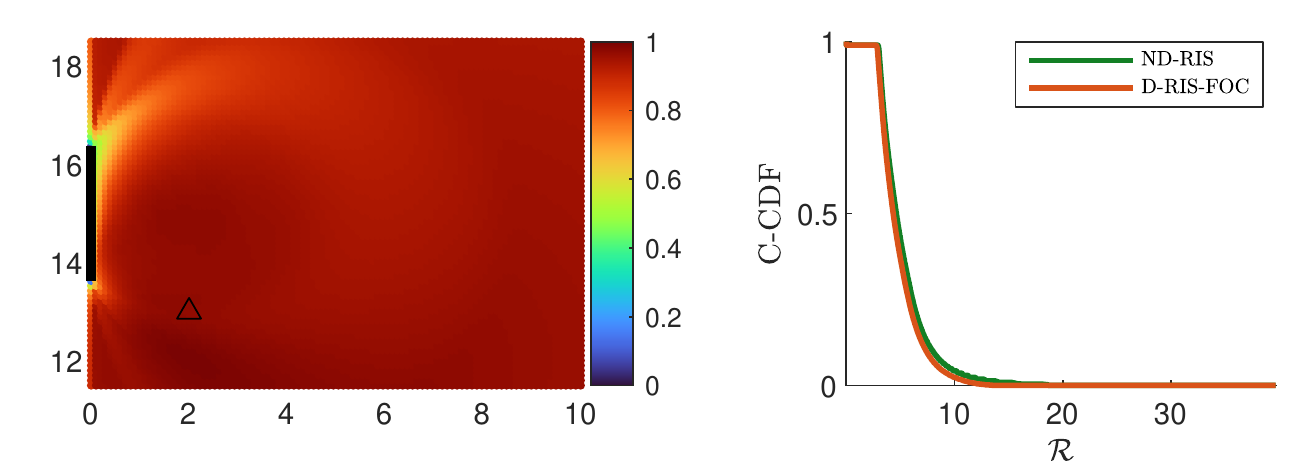}\label{SubFig:rate_map_BS_x020y130}}\\
	\subfloat[width=\columnwidth][Mid-point of the MIMO transmitter $\mathbf{c}_{\rm B}=(2.0,14.5,0.0)^T$]{\includegraphics[width=\columnwidth]{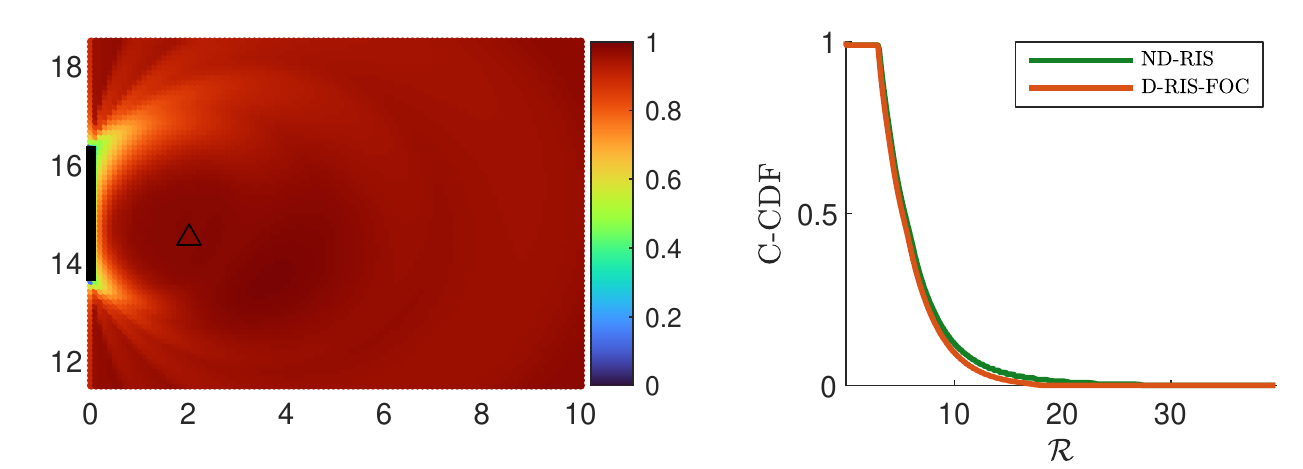}\label{SubFig:rate_map_BS_x020y145}}\\
	\subfloat[width=\columnwidth][Mid-point of the MIMO transmitter $\mathbf{c}_{\rm B}=(1.0,14.5,0.0)^T$]{\includegraphics[width=\columnwidth]{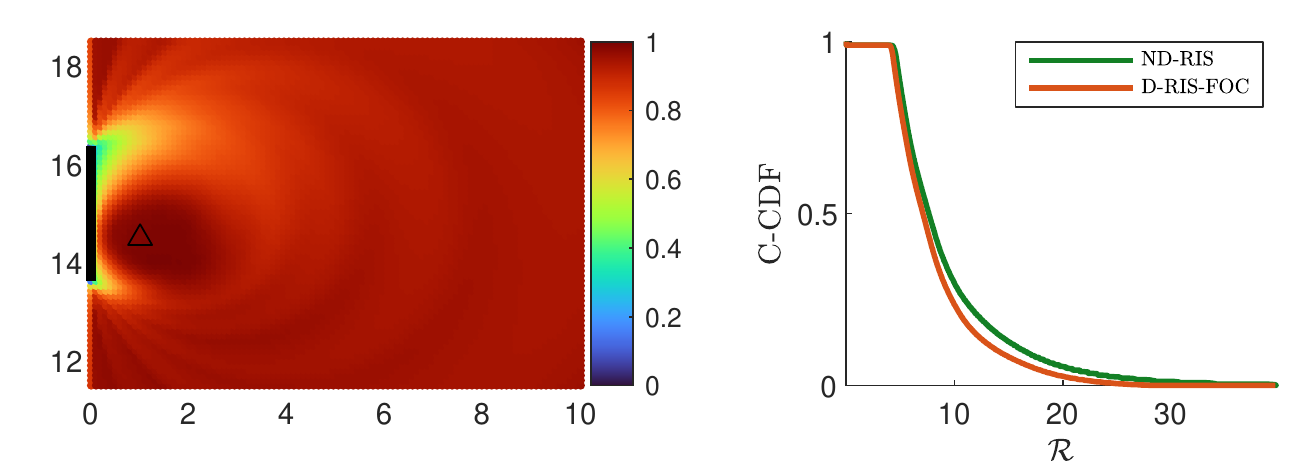}\label{SubFig:rate_map_BS_x010y145}}
		\caption{(left) Ratio of the achievable rate between the ``Non-diagonal RIS'' (ND-RIS) and the ``Diagonal RIS - Proposed'' (D-RIS-FOC) schemes as a function of the position of the MIMO receiver (a given point on the color map) for a fixed location of the MIMO transmitter. (right) Complementary cumulative distribution function (C-CDF) of the achievable rate for the ``Non-diagonal RIS'' (ND-RIS) and the ``Diagonal RIS - Proposed'' (D-RIS-FOC) schemes across the considered areas.}
	\label{MultiFig:RateMaps}
\end{figure*}
\begin{figure*}[!t]
	\centering
	\subfloat[width=\columnwidth][Mid-point of the MIMO transmitter $\mathbf{c}_{\rm B}=(2.0,13.0,0.0)^T$]{\includegraphics[width=\columnwidth]{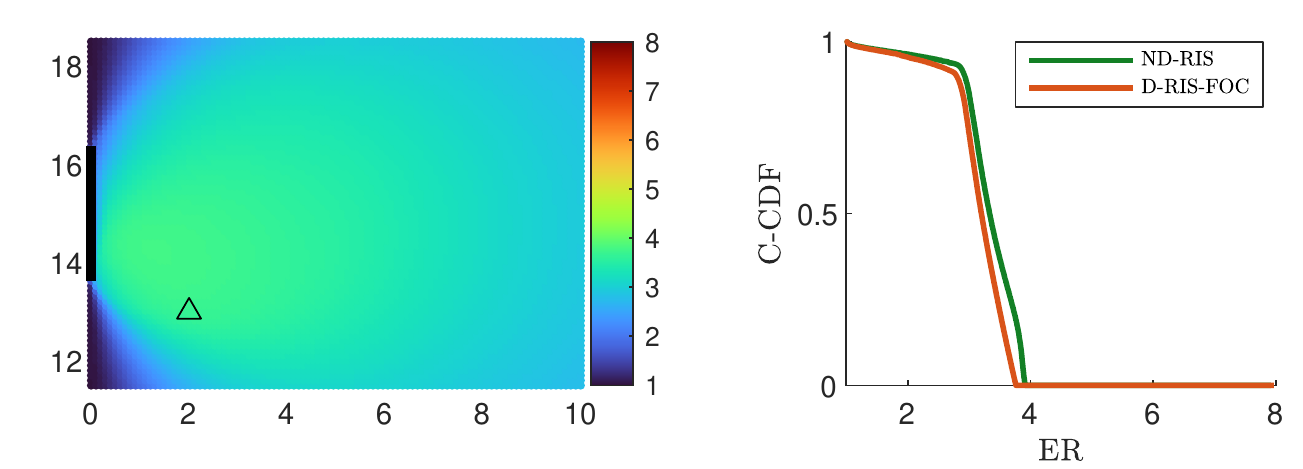}\label{SubFig:effR_map_BS_x020y130}}\\
	\subfloat[width=\columnwidth][Mid-point of the MIMO transmitter $\mathbf{c}_{\rm B}=(2.0,14.5,0.0)^T$]{\includegraphics[width=\columnwidth]{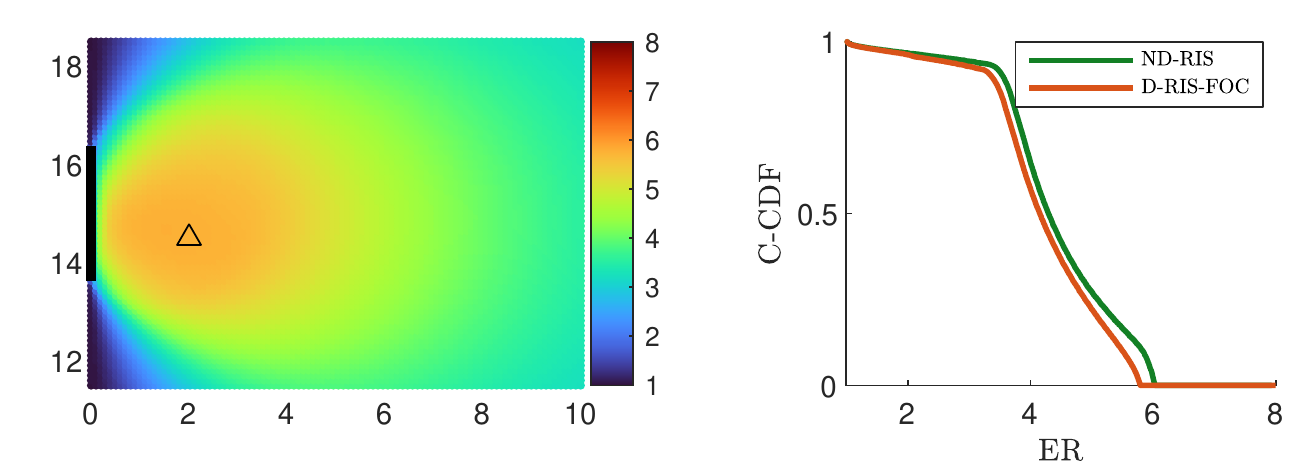}\label{SubFig:effR_map_BS_x020y145}}\\
	\subfloat[width=\columnwidth][Mid-point of the MIMO transmitter $\mathbf{c}_{\rm B}=(1.0,14.5,0.0)^T$]{\includegraphics[width=\columnwidth]{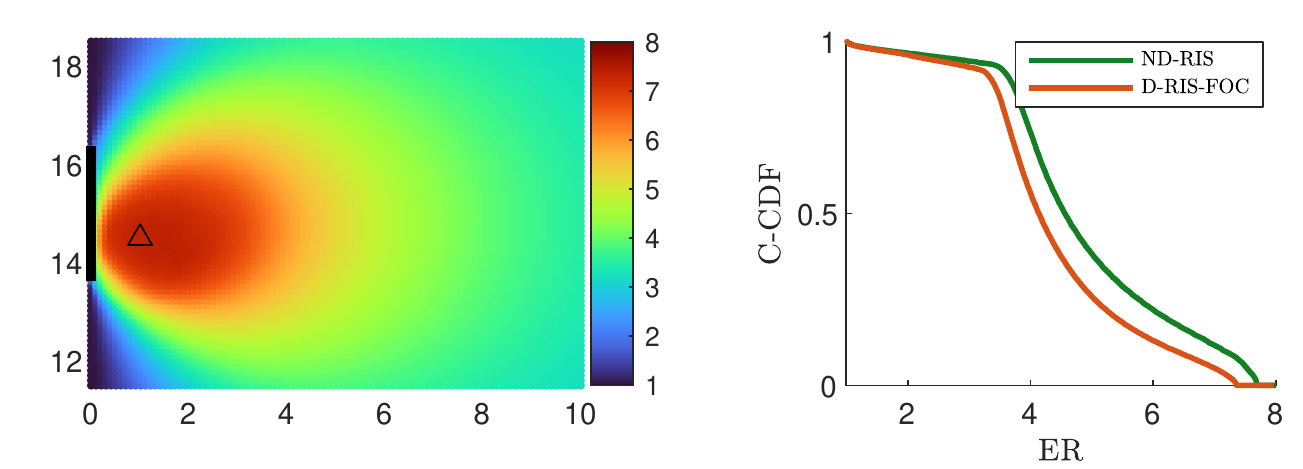}\label{SubFig:effR_map_BS_x010y145}}
		\caption{(left) Effective rank (ER) of the ``Diagonal RIS - Proposed'' (D-RIS-FOC) scheme as a function of the position of the MIMO receiver (a given point on the color map) for a fixed location of the MIMO transmitter. (right) C-CDF of the ER for the ``Non-diagonal RIS'' (ND-RIS) and the ``Diagonal RIS - Proposed'' (D-RIS-FOC) schemes across the considered areas.}
	\label{MultiFig:EffRMaps}
\end{figure*}
\afterpage{\clearpage}
\subsection{Non-Paraxial Setup}
In this sub-section, we analyze the effectiveness of the proposed approximated diagonal design for RIS in non-paraxial setups. To this end, we compute and compare the achievable rates and the effective rank \cite{DBLP:conf/eusipco/RoyV07} provided by the ``Non-diagonal RIS'' (ND-RIS) and ``Diagonal RIS - Proposed'' (D-RIS-FOC) schemes, as a function of the location of the MIMO receiver, while keeping the location of the MIMO transmitter fixed. The effective rank, in particular, provides information on the number of degrees of freedom that is achievable in the considered channel model.

In Fig. \ref{MultiFig:RateMaps}, we show the ratio of the achievable rates between the ``Diagonal RIS - Proposed'' (D-RIS-FOC) and the ``Non-diagonal RIS'' (ND-RIS) schemes. By definition, the ratio lies in the range $[0, 1]$, and the closer to one the ratio is the closer to the optimum the ``Diagonal RIS - Proposed'' (D-RIS-FOC) scheme is. The location of the MIMO transmitter is denoted by a black triangle, and the RIS is represented by a black segment. The MIMO transmitter, RIS, and MIMO receiver are assumed to be parallel to each other, i.e., $\psi_{\rm F}=\psi_{\rm B}=0$. We see that the proposed approach provides rates very close to the optimum, except for some locations that are very close to the RIS (within one meter from the RIS). In addition, we show the complementary cumulative distribution function of the achievable rates across the considered areas. The obtained curves confirm the good performance offered by the ``Diagonal RIS - Proposed'' (D-RIS-FOC) scheme with respect to the optimum. As expected, the closer to the RIS the MIMO transmitter, the better the rate \cite{DBLP:journals/corr/abs-2104-12108}.

In Fig. \ref{MultiFig:EffRMaps}, we illustrate the effective rank provided by the ``Diagonal RIS - Proposed'' (D-RIS-FOC) scheme across the considered area. We see that multiple transmission modes are available for different locations of the MIMO receiver, and the effective rank depends on the location of the MIMO transmitter. The obtained results confirm that the network topology and the locations of the MIMO transmitter, RIS, and MIMO receiver are essential parameters for ensuring multi-mode communications. Also, we compare the effective rank provided by the ``Diagonal RIS - Proposed'' (D-RIS-FOC)  scheme against the ``Non-diagonal RIS'' (ND-RIS) scheme. We note that the proposed diagonal design, given in a closed-form expression, provides an effective rank that is close to the optimum, but at a lower implementation complexity.

\begin{figure}[!t]
     \centering
         \includegraphics[width=\columnwidth]{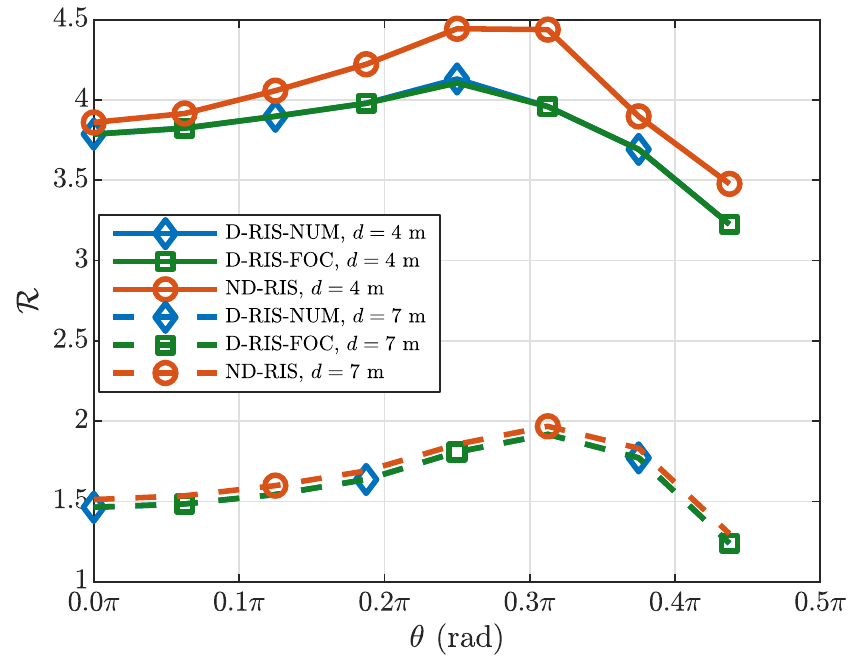}
         \caption{Comparison of the achievable rates offered by the ``Non-diagonal RIS'' (ND-RIS), ``Diagonal RIS - Numerical'' (D-RIS-NUM), and ``Diagonal RIS - Proposed'' (D-RIS-FOC) schemes as a function of the angle $\theta_{\rm F}=-\theta_{\rm B}=\theta$, while keeping the distances $d_{\rm{F}}=d_{\rm{B}}=d$ fixed (solid lines: $d=4$ meters; dashed lines: $d=7$ meters). The MIMO transmitter, RIS, and MIMO receiver are parallel to each other ($\psi_{\rm F}=\psi_{\rm B}=0$).}
         \label{fig:var_theta}
\end{figure}
In Fig. \ref{fig:var_theta}, we compare the achievable rates of the three considered schemes as a function of $\theta_{\rm F}=-\theta_{\rm B}=\theta$, while keeping the distances $d_{\rm{F}}=d_{\rm{B}}=d$ fixed. Overall, we see that the proposed approach provides rates that are virtually identical to those obtained with complex numerical algorithms and slightly worse than those obtained by using non-diagonal RIS only in some cases.

\begin{figure}[!t]
     \centering
         \includegraphics[width=\columnwidth]{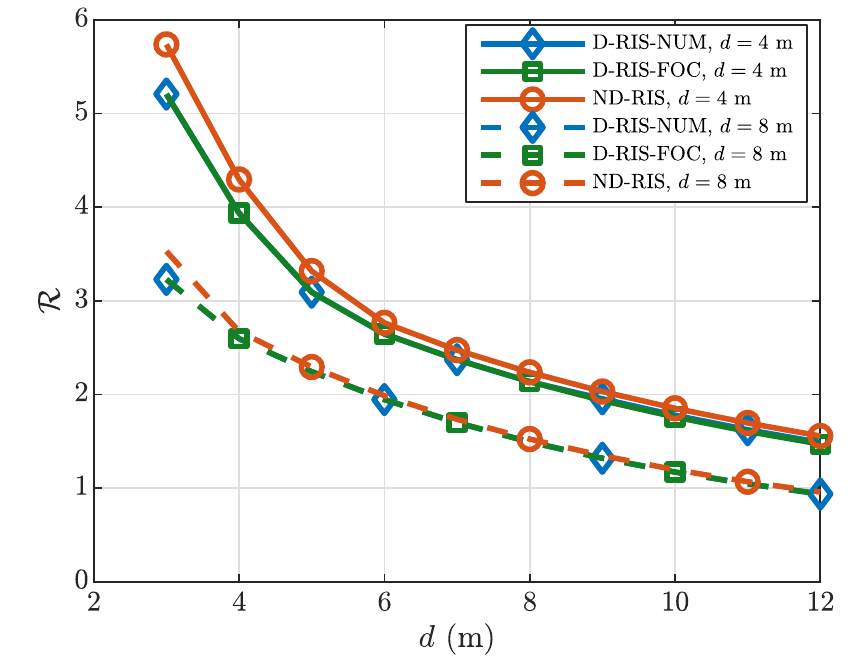}
         \caption{Comparison of the achievable rates offered by the ``Non-diagonal RIS'' (ND-RIS), ``Diagonal RIS - Numerical'' (D-RIS-NUM), and ``Diagonal RIS - Proposed'' (D-RIS-FOC) schemes as a function of the distance $d_{\rm{B}}=d$, while keeping the distance $d_{\rm{F}}$ (solid lines: $d_{\rm{F}}=4$ meters; dashed lines: $d_{\rm{F}}=8$ meters), and the angles $\theta_{\rm F}= \pi/8$ and $\theta_{\rm B} = 3\pi/8$ fixed. The MIMO transmitter, RIS, and MIMO receiver are parallel to each other ($\psi_{\rm F}=\psi_{\rm B}=0$).}
         \label{fig:var_rho}
\end{figure}
In Fig. \ref{fig:var_rho}, we compare the achievable rates offered by the considered three schemes as a function of the distance $d_{\rm{B}}=d$, while keeping the other parameters fixed. The conclusions are similar to those drawn from the analysis of Fig. \ref{fig:var_theta}.

\begin{figure*}[!t]
	\centering
	\subfloat[][$\psi_{\rm B}=-\pi/8$]{\includegraphics[width=0.33\columnwidth]{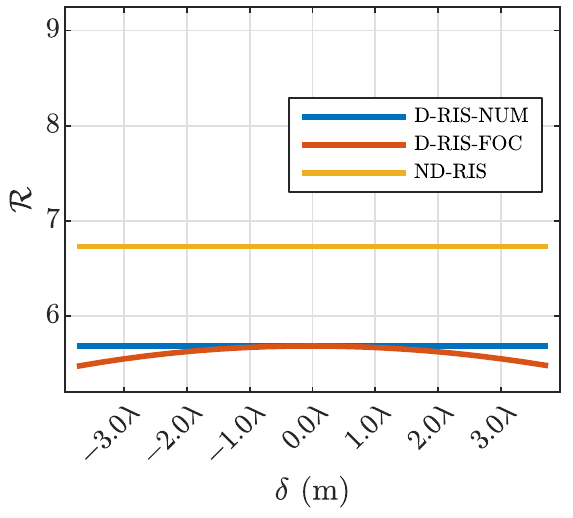}\label{SubFig:var_focusing_UE_1}}
	\subfloat[][$\psi_{\rm B}=\pi/4$]{\includegraphics[width=0.33\columnwidth]{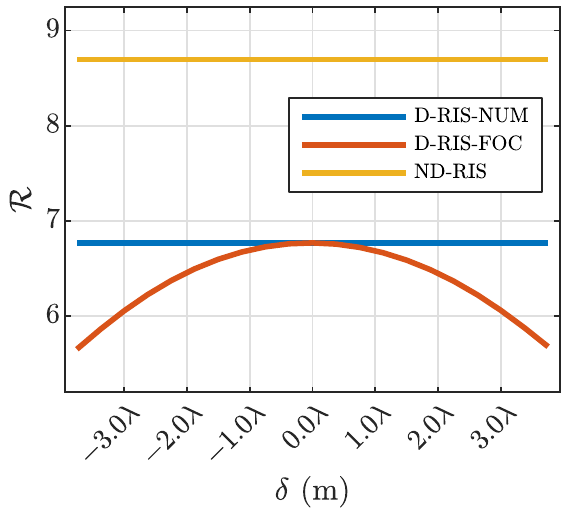}\label{SubFig:var_focusing_UE_2}}
	\subfloat[][$\psi_{\rm B}=\pi/8$]{\includegraphics[width=0.33\columnwidth]{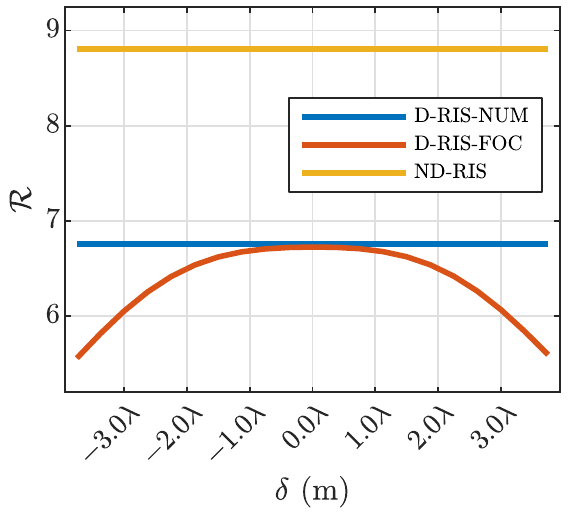}\label{SubFig:var_focusing_UE_3}}
		\caption{Comparison of the achievable rates offered by the ``Non-diagonal RIS'' (ND-RIS), ``Diagonal RIS - Numerical (D-RIS-NUM)'', and ``Diagonal RIS - Proposed (D-RIS-FOC)'' schemes as a function of the rotation angle $\psi_{\rm B}$ and the focusing point $\bf{p}_{\rm{RF}}$. Specifically, $\bf{p}_{\rm{RF}}=\delta\frac{\mathbf{u}_{M}-\mathbf{c}_{\rm U}}{\|\mathbf{u}_{M}-\mathbf{c}_{\rm U}\|}$, where $0 \le \delta \le 1$ and $\mathbf{u}_{M}$ is the $M$th element of the MIMO receiver. Setup: $\theta_{\rm F}= \pi/4$, $\theta_{\rm B}= \pi/3$, $d_{\rm{F}}=2$ meters, $d_{\rm{B}}=3.5$ meters, $\psi_{\rm F}=0$, and $\bf{p}_{\rm{RB}} = \mathbf{c}_{\rm B}$.}
	\label{MultiFig:var_focusing_UE}
\end{figure*}
In Fig. \ref{MultiFig:var_focusing_UE}, we compare the achievable rates offered by the considered schemes as a function of the rotation angle $\psi_{\rm B}$ and the focusing point $\bf{p}_{\rm{RF}}$. The main objective is to understand the impact of rotating the receiver, which results in large deviations from the typical paraxial setup, and to evaluate the impact of the focusing points at the MIMO receiver, i.e., whether choosing the mid-point as the focusing point, as dictated by the paraxial setup, is a good choice. The obtained results unveil two main performance trends: (1) Selecting the mid-point of the MIMO receiver as the focusing point is the best choice, as it provides rates that are very close to those obtained by using the ``Diagonal RIS - Numerical'' (D-RIS-NUM) scheme; and (2) rotating the MIMO receiver increasese the gap between the rates attained by the ``Diagonal RIS - Proposed'' (D-RIS-FOC) scheme and the ``Non-diagonal RIS'' (ND-RIS) scheme. In general, a non-diagonal design for RIS provides better rates in the presence of rotations.

\section{Conclusion}
\label{Sec:Conclusion}
In this paper, we have analyzed the spatial multiplexing gains of RIS-aided MIMO channels in the near field. We have proved that the best design for nearly-passive RIS results in a non-diagonal matrix of reflection coefficients. Due to the non-negligible complexity of non-diagonal designs for RIS, we have proposed a closed-form diagonal design that is motivated and is proved to be optimal, from the end-to-end channel capacity standpoint, in line-of-sight channels and when the MIMO transmitter, RIS, and MIMO receiver are deployed according to the paraxial setup. In different network typologies and over fading channels, the proposed design is sub-optimal. However, extensive simulation results in line-of-sight (free-space) channels have confirmed that it provides good performance in non-paraxial setups as well. Specifically, we have shown that the proposed diagonal design provides rates that are close to those obtained by numerically solving non-convex optimization problems at a high computational complexity, as well as to those attained, in several considered network setups, by capacity-achieving non-diagonal RIS designs.

\appendices

\section{Proof of Proposition 1}\label{App:Capacity}
In this appendix, we provide a formal proof for Proposition 1. The proposed approach generalizes, to RIS-aided MIMO channels, the proof originally given in \cite{DBLP:journals/ett/Telatar99} for MIMO channels. The proof leverages the assumption that the matrix of reflection coefficients of the RIS is unitary but is not necessarily diagonal.

By definition, assuming that the channels are fixed and given, the mutual information is
\begin{IEEEeqnarray}{rCl}
    I(\B{x},\B{y}) & = & \log\det\left(\B{I}_M+\Hphi\B{Q}\Hphi^{\herm}\right) \\ & = & \log\det\left(\B{I}_M+\B{Q}\Hphi^{\herm}\Hphi\right)
\end{IEEEeqnarray}
where the last equality follows by using the identity $\det \left( {{\bf{I}} + {\bf{AB}}} \right) = \det \left( {{\bf{I}} + {\bf{BA}}} \right)$.

The channel product $\Hphi^{\herm}\Hphi$ can be expressed as
\begin{IEEEeqnarray}{rCl}
    \Hphi^{\herm}\Hphi&=&(\Hb\PHI\Hf)^{\herm}(\Hb\PHI\Hf)\nonumber\\
    &=&\left(\Hf^{\herm}\PHI^{\herm}\Hb^{\herm}\right)\left(\Hb\PHI\Hf\right)\nonumber\\
    &=&\left(\Uf\Sf\Vf^{\herm}\right)^{\herm}\PHI^{\herm}\left(\Vb\Sb^{\herm}\Sb\Vb^{\herm}\right)\PHI\left(\Uf\Sf\Vf^{\herm}\right)
\end{IEEEeqnarray}

Similarly to \cite{DBLP:journals/twc/TangH07}, without loss of generality, we can express the matrix $\PHI$ as
\begin{equation}
    \PHI = \Vb\B{X}\Uf^{\herm}
\end{equation}
with $\B{X}$ being a generic unitary matrix, not necessarily diagonal.

Therefore, the mutual information can be expressed as
\begin{IEEEeqnarray}{rCl}
    I(\B{x},\B{y})&=&\log\det\left(\B{I}_M+\B{Q}\Vf\Sf^{\herm}\B{X}^{\herm}\Sb^{\herm}\Sb\B{X}\Sf\Vf^{\herm}\right)\\
    &=&\log\det\left(\B{I}_M+\Vf^{\herm}\B{Q}\Vf\Sf^{\herm}\B{X}^{\herm}\Sb^{\herm}\Sb\B{X}\Sf\right)\\
    &=&\log\det\left(\B{I}_M+\T{\B{Q}}\T{\B{X}}\right)
\end{IEEEeqnarray}
where the last equality is obtained by defining
\begin{IEEEeqnarray}{l}
    \T{\B{Q}}=\Vf^{\herm}\B{Q}\Vf    \label{eq:Q_tilde}\\
    \T{\B{X}}=\Sf\B{X}^{\herm}\Sb\Sb\B{X}\Sf
\end{IEEEeqnarray}

Let us focus on the matrix $\T{\B{Q}}$. Since $\Vf$ is unitary, we have $\operatorname{tr}(\T{\B{Q}})=\operatorname{tr}(\B{Q})$. Therefore, maximizing over $\B{Q}$ is equivalent to maximizing over $\T{\B{Q}}$, without changing any constraints of the original optimization problem, since $\T{\B{Q}}$ is positive semi-definite if $\B{Q}$ is positive semi-definite. This is proven in the following lemma.
\begin{lemma}\label{lemma:Q_tilde}
Matrix $\T{\B{Q}}$ is positive semi-definite.
\end{lemma}
\begin{proof}

This lemma can be proven by observing the following.
\begin{itemize}
    \item $\T{\B{Q}}$ is an Hermitian matrix, since 
    \begin{equation}
    \T{\B{Q}}^{\herm}=(\Vf^{\herm}\B{Q}\Vf)^{\herm}=\Vf^{\herm}\B{Q}\Vf=\T{\B{Q}}
\end{equation}
    
    \item As $\T{\B{Q}}$ is an Hermitian matrix, it is positive semi-definite if $\B{z}^{\herm}\T{\B{Q}}\B{z}\geq0$ for any complex vector $\B{z}$. This follows because 
        \begin{equation}
    \B{z}^{\herm}\T{\B{Q}}\B{z}=\B{z}^{\herm}\Vf^{\herm}\B{Q}\Vf\B{z}=\T{\B{z}}^{\herm}\B{Q}\T{\B{z}} \ge 0
\end{equation}
where the last inequality follows because $\B{Q}$ is positive semi-definite and $\T{\B{x}}=\Vf\B{z}$ is a generic complex vector.

\end{itemize}

This completes the proof.
\end{proof}

The mutual information can be rewritten as
\begin{IEEEeqnarray}{rCl}
    I(\B{x},\B{y})&=&\log\det\left(\B{I}_M+\Sf\Vf^{\herm}\B{Q}\Vf\Sf^{\herm}\B{X}^{\herm}\Sb^{\herm}\Sb\B{X}\right)\nonumber\\
    &=&\log\det\left(\B{I}_M+\D{\B{Q}}\,\D{\B{X}}\right)
\end{IEEEeqnarray}
where the last equality is obtained by defining
\begin{IEEEeqnarray}{l}
    \D{\B{Q}}=\Sf\Vf^{\herm}\B{Q}\Vf\Sf^{\herm}=\Sf\T{\B{Q}}\Sf^{\herm}\label{eq:dashed_Q}\\
    \D{\B{X}}=\B{X}^{\herm}\Sb^{\herm}\Sb\B{X}
\end{IEEEeqnarray}

Let us consider the matrix $\D{\B{X}}$. It can rewritten as
\begin{IEEEeqnarray}{rCl}
    \D{\B{X}}&=&\B{X}^{\herm}\T{\B{\Sigma}}_{\B{B}}\B{X}\nonumber\\
    &=&\left(\T{\B{\Sigma}}_{\B{B}}^{1/2}\B{X}\right)^{\herm}\left(\T{\B{\Sigma}}_{\B{B}}^{1/2}\B{X}\right)
\end{IEEEeqnarray}
where $\T{\B{\Sigma}}_{\B{B}} = \Sb^{\herm}\Sb$ is a square diagonal matrix with non-negative real entries on its diagonal. The last equality follows because the matrix $\T{\B{\Sigma}}_{\B{B}}$ is invertible and can written as $\T{\B{\Sigma}}_{\B{B}} ={\B{\Sigma}}_{\B{B}}^{1/2} {\B{\Sigma}}_{\B{B}}^{1/2}$. Also, the last equality implies that $\D{\B{X}}$ is a positive semi-definite matrix for any choice of $\B{X}$.

Accordingly, $\D{\B{X}}$ can be expressed in terms of its eigenvalue decomposition as $\D{\B{X}}=\B{U}_{\D{\B{X}}}\B{\Lambda}_{\D{\B{X}}}\B{U}_{\D{\B{X}}}^{\herm}$. As a result, the mutual information can be rewritten as
\begin{IEEEeqnarray}{rCl}
    I(\B{x},\B{y})&=&\log\det\left(\B{I}_M+
    \D{\B{Q}}
    \left(\B{U}_{\D{\B{X}}}\B{\Lambda}_{\D{\B{X}}}\B{U}_{\D{\B{X}}}^{\herm}\right)
    \right)\nonumber\\
    &=&\log\det\left(\B{I}_M+
    \B{\Lambda}_{\D{\B{X}}}^{1/2}\B{U}_{\D{\B{X}}}^{\herm}\D{\B{Q}}
    \B{U}_{\D{\B{X}}}\B{\Lambda}_{\D{\B{X}}}^{1/2}
    \right) \label{eq:I_psd}
\end{IEEEeqnarray}

From \eqref{eq:dashed_Q} and \eqref{eq:I_psd}, the following considerations can be made.
\begin{itemize}
    \item The matrix $\D{\B{Q}}=\Sf\T{\B{Q}}\Sf^{\herm}$ is positive semi-definite, since $\T{\B{Q}}$ is positive semi-definite. The proof is the same as for Lemma~\ref{lemma:Q_tilde}.
    
    \item The matrix $\B{U}_{\D{\B{X}}}^{\herm}\D{\B{Q}}\B{U}_{\D{\B{X}}}$ is positive semi-definite, since $\D{\B{Q}}$ is positive semi-definite. The proof is the same as for Lemma~\ref{lemma:Q_tilde}.
    
    \item The matrix $\B{\Lambda}_{\D{\B{X}}}^{1/2}\B{U}_{\D{\B{X}}}^{\herm}\D{\B{Q}}\B{U}_{\D{\B{X}}}\B{\Lambda}_{\D{\B{X}}}^{1/2}$ is positive semi-definite, since $\B{U}_{\D{\B{X}}}^{\herm}\D{\B{Q}}\B{U}_{\D{\B{X}}}$ is positive semi-definite. The proof is the same as for Lemma~\ref{lemma:Q_tilde}.
\end{itemize}

Therefore, the matrix $\B{I}_M+\B{\Lambda}_{\D{\B{X}}}^{1/2}\B{U}_{\D{\B{X}}}^{\herm}\D{\B{Q}}\B{U}_{\D{\B{X}}}\B{\Lambda}_{\D{\B{X}}}^{1/2}$ is positive semi-definite as well. From \cite{DBLP:journals/ett/Telatar99}, it is known that, for any generic positive semi-definite matrix $\B{A}$ (therein referred to as non-negative defined), we have
\begin{equation}
    \det\left(\B{A}\right) \leq \prod_i \B{A}\left(i,i\right) =  \prod_i \left(\B{A}\right)_{i,i} \label{eq:UpperBound}
\end{equation}
where the notation $\B{A}\left(i,i\right)=\left(\B{A}\right)_{i,i}$ is employed to simplify the writing.

Let us consider the formulation of the mutual information in (33), which is the most convenient one for further analysis. By using the upper-bound in \eqref{eq:UpperBound}, we obtain
\begin{IEEEeqnarray}{rCl}
    I(\B{x},\B{y})&=&\log\det\left(\B{I}_M+\T{\B{Q}}\T{\B{X}}\right)\nonumber\\
    &\leq&\log\prod_i\left(\B{I}_M+\T{\B{Q}}\T{\B{X}}\right)_{ii}
\end{IEEEeqnarray}

Based on \cite{DBLP:journals/ett/Telatar99}, a sufficient condition to attain the upper-bound, i.e., (44) is fulfilled with equality, is that $\T{\B{X}}$ and $\T{\B{Q}}$ are two diagonal matrices. Under these assumptions, we obtain
\begin{IEEEeqnarray}{rCl}
    I(\B{x},\B{y})&=&\log\prod_i\left(\B{I}_M+\T{\B{Q}}\T{\B{X}}\right)_{ii}\nonumber\\
    &=&\log\prod_i\left(\B{I}_M+\T{\B{Q}}\Sf\B{X}^{\herm}\Sb\Sb\B{X}\Sf\right)_{ii}\nonumber\\
    &=&\log\prod_i\left(1+\T{\B{Q}}\left(i,i\right)\Sfi\B{X}^{\herm}\left(i,i\right)\Sbi\Sbi\B{X}\left(i,i\right)\Sfi\right)\nonumber\\
    &=&\log\prod_i\left(1+\Sbi^2\Sfi^2\T{\B{Q}}\left(i,i\right)\right)\label{eq:DiagChannel}
\end{IEEEeqnarray}
where the last equality follows because $\B{X}^{\herm}\left(i,i\right)\B{X}\left(i,i\right)=1$, since $\B{X}$ is a unitary and diagonal if $\T{\B{X}}$ is diagonal for achieving the upper-bound.

In summary, we have proved that $\T{\B{Q}}=\Vf^{\herm}\B{Q}\Vf$ and that $\PHI = \Vb\B{X}\Uf^{\herm}$, for any unitary and diagonal matrix $\B{X}$, is capacity-achieving. Without loss of generality, therefore, we can consider $\B{X} = {\B{I}}_K$. The proof follows by noting that the capacity is achieved by applying the water-filling power allocation to the ordered product of the singular values of the transmitter-RIS and RIS-receiver channels.

\bibliographystyle{IEEEtran}
\bibliography{IEEEabrv,biblio}

\end{document}